\documentclass[%
 aip,
 amsmath,amssymb,
preprint,
reprint, onecolumn %%%%% Comment onecolumn to override one column format %%%%%%%
]{revtex4-1}
\usepackage[english]{babel}

\usepackage{graphicx}

\usepackage{amsmath}
\usepackage{amsthm}
\usepackage{amssymb}
\usepackage{fancyhdr}
\usepackage{lastpage}
\usepackage{multirow}
\usepackage{color}
\usepackage{ifthen}
\usepackage{enumitem}
\usepackage{ stmaryrd }
\usepackage{ bbold }
\usepackage{subcaption}

\usepackage[dvipsnames]{xcolor}

\newcommand{\R}{\mathbb{R}}
\newcommand{\N}{\mathbb{N}}

\newtheorem{theorem}{Theorem}

\newtheorem{prop}[theorem]{Proposition}
\newtheorem{lemma}[theorem]{Lemma}
\newtheorem{cor}[theorem]{Corollary}

\begin{document}

%\preprint{}

\title{Cluster Synchronization via Graph Laplacian Eigenvectors}



\author{Tobias Timofeyev}
 \email{tobias.timofeyev@uvm.edu}
\author{Alice Patania}%
 \altaffiliation[Also at ]{Vermont Complex Systems Institute, University of Vermont, Burlington, Vermont, USA.}%Lines break automatically or can be forced with \\
\affiliation{%
 Department of Mathematics and Statistics\\
 University of Vermont, Burlington, Vermont, USA.
}%
%\thanks{}
% \affiliation{}


\date{\today}

\begin{abstract}
Almost equitable partitions (AEPs) have been linked to cluster synchronization in oscillatory systems, highlighting the importance of structure in collective network dynamics. 
We provide a general spectral framework that formalizes this connection, showing how eigenvectors associated with AEPs span a subspace of the Laplacian spectrum that governs partition-induced synchronization behavior. This offers a principled reduction of network dynamics, allowing clustered states to be understood in terms of quotient graph projections. 
Our approach clarifies the conditions under which transient hierarchical clustering and multi-frequency synchronization emerge, and connects these dynamical phenomena directly to network symmetry and community structure. 
In doing so, we bridge a critical gap between static topology and dynamic behavior-namely, the lack of a spectral method for analyzing synchronization in networks that exhibit exact or approximate structural regularity.
Perfect AEPs are rare in real-world networks since most have some degree of irregularity or noise. We define a relaxation of an AEP we call a quasi-equitable partition at level $\delta$ ($\delta-$QEP). $\delta-$QEPs can preserve many of the clustering-relevant properties of AEPs while tolerating structural imperfections and noise. This extension enables us to describe synchronization behavior in more realistic scenarios, where ideal symmetries are rarely present.
Our findings have important implications for understanding synchronization patterns in real-world networks, from neural circuits to power grids.
\end{abstract}

%\keywords{}

\maketitle


The Kuramoto model is a network oscillator model useful for describing collective synchronization over a network of interacting entities, with applications from neuroscience to engineering to the social sciences. In this work we explore the relationship between the structure of the underlying network and the clustered synchronization that can occur in the dynamics. We look at the underlying network structure through the lens of Almost Equitable Partitions (AEPs), which are defined by regularity in their connectivity and relate to common notions of network community structure.
The existence of this kind of partition is reflected in the eigenvectors of the network’s matrix representation as a graph Laplacian. We use this to relate the theory of community detection and symmetries of networks to the synchronization dynamics of the Kuramoto model by building a description in terms of these eigenvectors. When treated in this way, the question of collective synchronization becomes one of the population dynamics of eigenvectors. This relationship allows us to shed light on transient cluster synchronization, multi frequency clustering, and the conditions under which they can occur.

\section{Introduction\label{intro}}
The study of complex networks and their dynamical behaviors has been a fertile ground for understanding synchronization phenomena across different fields. From meso-scale neural activation, to power flow networks, numerous complex networks exhibit interacting oscillatory behavior. In the study of these systems, it is often natural to ask what sort of synchronization behaviors can occur. In recent years, there has been evidence from neuroscience of an unexplored relationship between resting state clusters in the brain and the Laplacian eigenvectors of functional connectivity networks.\cite{atasoy_human_2016} While existing studies increasingly recognize the role of non-pairwise interactions in shaping complex neuronal synchrony,\cite{majhi_patterns_2025} they fall short of explaining the correspondence between functional clustering and structurally defined eigenvectors, particularly in the presence of high amounts of noise. 
Our work approaches this question from the perspective of network oscillatory models, closing this gap between structural and functional clusters. 

Pioneering work by Kuramoto\cite{Kuramoto_chem_oscillations} on what is now called the Kuramoto model, brought fundamental insights into how coupled oscillators spontaneously synchronize, and laid the groundwork for subsequent investigations into network-level synchronization dynamics.\cite{strogatz_kuramoto_2000} The Kuramoto model and its derivatives are used as an abstraction of coupled periodic behavior for study. This model represents vertices as oscillators and couples their differential equations over the edges of the network. Consider a weighted graph $G$ with adjacency matrix $A$. Let $\sigma\in \R_{\geq 0}$ be a global coupling constant, and for every vertex $i$, define the natural frequency as $\omega_i\in \R$. Then the phase of vertex $i$ over time is determined by the differential equation
\begin{equation}\label{Kuramoto}
   \dot{\theta_i} = \omega_i - \sigma \sum_{j} A_{ij} \sin(\theta_i-\theta_j).
\end{equation} 

The synchronization dynamics of coupled oscillator models can take many forms. In this paper we explore full synchronization and clustered dynamics in the context of the Kuramoto model, but in more general coupled oscillator models, there exist a broad range of hybrid synchronization patterns, as reviewed by Parastesh et al.,\cite{parastesh2021chimeras} studying chimera states emerging from structurally heterogeneous or imperfectly symmetric networks. In parallel, studies have also demonstrated that cluster synchronization can arise not only from network symmetries but also from nonlinear coupling effects. For instance, Skardal et al.\cite{skardal_cluster_2011} and Gong and Pikovsky\cite{gong_low-dimensional_2019} have shown that higher-order harmonic terms in the coupling function can induce multi-cluster dynamics even in fully connected networks without structural heterogeneity. 

Naturally, the synchronization characteristics of a coupled oscillator model, and the Kuramoto model specifically, are informed by the structure of the graph it is placed upon. However, the nature of this relationship is unclear, particularly due to the multitudinous structures a graph can take. 
Recent advances in the study of coupled oscillator dynamics have progressively refined our understanding of how graph structure influences synchronization potential\cite{nguyen2025communities, schaub_graph_2016, kato_cluster_2023} and have proven that Almost Equitable Partitions (AEPs) facilitate cluster synchronization via their quotient graph. Specifically, the characterization of synchronized clusters in networked oscillators has highlighted the intricate relationship between graph structure and dynamical coherence.\cite{abiad_characterizing_2022, OClery_Observability_2013}
Simultaneously, in more general signal processing literature, a novel approach has been emerging, which focuses on the discretization of the Fourier analysis framework to better study the fundamental tradeoff between a signal's distribution across the vertices of a graph and in its spectral domain.\cite{shuman_emerging_2013}

Despite these significant advances, a crucial gap persists between the spectral properties of AEPs, studied in graph theory and network science, and the cluster synchronization they allow in complex networks. Traditionally, their spectral properties have proven to be crucial for the detectability of AEPs (community structure) on graphs,\cite{zhang_almost_2021, krzakala_spectral_2013} but have till this point found little purchase in the space of  signals constrained to the vertices of graphs and importantly for us, the dynamics of cluster synchronization.
Foundational studies by Schaub et al.,\cite{schaub_graph_2016} O'Clery et al.\cite{OClery_Observability_2013} and Kato et al.\cite{kato_cluster_2023} have pointed to the importance of network structure and specifically Almost Equitable Partitions in the formation of clustered dynamics.
However, these dynamical insights on the structurally equitable or near-equitable configurations have yet to be tied into the corresponding spectral theory of these structures for analysis.

Building upon spectral graph theory and drawing from recent developments in network synchronization research, we propose a generalized framework that extends traditional equitable partition analysis. Researchers like Schaub et al., and Thibeault et al.\cite{schaub_hierarchical_2023, thibeault_threefold_2020,thibeault_low-rank_2024} have provided foundational insights into matrix reduction and system dimensionality that inform our approach.
At the same time, symmetry-based approaches\cite{cho_stable_2017, zhang_symmetry-independent_2020, pecora_cluster_2014, khanra_identifying_2022} have highlighted how both explicit and hidden symmetries in network topology can dictate the formation and stability of synchronized clusters. 

However, a unifying spectral framework that links approximate structural patterns to emergent dynamical behavior has been missing. Our work addresses this gap by showing that eigenvectors associated with quotient graphs induced by almost equitable partitions (AEPs) form a spectral subspace that governs the coarse-grained dynamics of clustered synchronization. By leveraging eigenvector interpretation techniques,\cite{shuman_emerging_2013} we extend this framework to define quasi-equitable partitions at level $\delta$ ($\delta$-QEPs), which may retain dynamical relevance even when exact AEP structure is disrupted by noise or irregularity.

The spectral approach introduced here provides new robust insights into the relationship between quasi-equitable partitioning and synchronized behavior, while also providing an easily generalizable method for analyzing synchronization potential in coupled oscillatory networks. We showcase the power of the proposed framework with two examples: a spectral analysis of multi-frequency clustering, and extensions of the Kuramoto model that include weighted edges and antisymmetric phase lag.

In the sections that follow, we build the necessary foundation for an understanding of cluster synchronization in the Kuramoto model, entirely through the information encoded in its graph Laplacian. We begin by considering the structural information within the Laplacian's spectrum, and particularly the role of partitions therein (Section 2). We then discuss the relationship between the spectrum of the Laplacian and its dynamics. Following this background, we present a novel spectrally motivated description of cluster synchronization, and discuss the stability of these properties with weighted graphs (Section 3). The presented spectral description of clustered dynamics allows for new and interesting applications. We explore the implications for multi-frequency clustering, and modeling with antisymmetric phase lag. Lastly, Section 5 offers concluding remarks and future research directions.

\section{The Spectral description of Synchronized Clusters}
\label{sec:spectral_description_of_synchronized_clusters}


The field of spectral graph theory involves itself with the linear algebra of graph representations and the relationship with graph structure. This crucial foundation for the field of network analysis allows us to leverage matrix algebra tools, built for computing, to understand graphs at scales we cannot feasibly study combinatorially.

Matrix representations of graphs come in many forms, and we present a few with relevance to our application. Let $G=(V,E,w)$ be a weighted connected graph with $n$ vertices $v_1, \dots, v_n\in V$ and $m=|E|$ edges $E \subset \{ \{u,v\} \mid u,v \in V \}$, where $w:E\to \R^+$ is the weight function for the edges. Perhaps the most natural description of this graph comes in the form of the adjacency matrix $A\in \R^{n\times n}$, whose $ij$th entry is $w(\{v_i,v_j\})$ if $\{v_i,v_j\}\in E$ and $0$ otherwise.

Another important representation comes in the form of the signed incidence matrix, also known as the boundary matrix. This matrix ignores weight assignments, and requires an arbitrary orientation on the edges of our graph. 
For simplicity, we induce an orientation from the vertex order. Define the signed incidence matrix as $B \in \R^{n\times m}$ such that $b_{ie}=1$ if edge $e=\{i,j\}$ and $i<j$, and $b_{ie}=-1$ if $i>j$. Finally, $b_{ie}=0$ if edge $e$ is not incident to vertex $i$. 
This boundary matrix serves as a fundamental building block of the spectral form of the Kuramoto equation shown later. We will see that when viewed as an operator, it offers a bridge between vertex values and edge values.

Finally we introduce the matrix we pay particular attention to in this paper, the graph Laplacian. We define the graph Laplacian as $L=D-A$, where $A$ is the adjacency matrix and $D\in \R^{n\times n}$ is the diagonal matrix of the row sums of $A$. The entries of $D$ are the sums of incident edge weights. In the case of uniform weight $1$ (an 'unweighted' graph), they equal the degrees of the corresponding vertices. The graph Laplacian is a cornerstone in spectral graph theory and has numerous applications in network analysis, clustering, and graph-based dynamical systems.\cite{Newman2018-ut}

One interesting feature of the graph Laplacian is in its relationship to the incidence matrix. Given our graph $G$, one can show that the Laplacian $L=BWB^T$ where $W\in \R^{m\times m}$ is the diagonal matrix of edge weights. 

The description of the vertex-wise Laplacian in terms of the boundary matrix gives rise to a dual matrix known as the weighted down edge Laplacian. This matrix is defined as $L^{DN} \equiv B^TBW$, and has eigenvectors in one to one correlation with that of the graph Laplacian. In particular, an eigenvector-value pair of the vertex-valued graph Laplacian, $(v_r,\lambda_r)$ has corresponding pair $(B^Tv_r,\lambda_r)$  of the edge-valued $L^{DN}$. Notice that the boundary matrix facilitates this duality between the vertex and edge spaces. This relationship is shown explicitly in appendix \ref{apdx:Coeff_Lapl}.

It is well-established that the eigenvectors of the Laplacian play a crucial role in revealing organizational structure in graphs, particularly in terms of their symmetry and connectivity patterns. For example, the multiplicity of the zeroth eigenvalue reveals the number of connected components and, in the case of multiplicity $1$, its corresponding normalized eigenvector is $v_0 = \frac{1}{\sqrt{N}}\mathbb{1}$. 

\subsection{Spectrum and Almost Equitable Partitions}\label{sec:graphMat}

One frame under which to understand symmetries of graphs or groupings of structurally similar vertices, is \textbf{almost equitable partitions} (AEP) of a graph, also known as externally equitable partitions. 

We define an AEP on an unweighted graph before proceeding to our weighted case. Let $N(v)$ denote the set of neighbors of $v\in V$. An AEP is a partition of the vertex set $\pi=\{V_1,\dots,V_k\}$ such that for $V_i, V_j\in \pi$ where $i\neq j$, there exists $d_{ij}\in \N$ such that for all $v\in V_i$, we have that $|N(v)\cap V_j|=d_{ij}$ (See figure \ref{fig:EXPOSITION} B.1). In other words, the partition cell a vertex is in determines the number of neighbors that vertex has in each other partition cell. Such a partition encodes a grouping of vertices that play the same `role' in the connectivity of the graph between the delineated partition cells. Note that the almost equitable partition ignores any connectivity within a partition cell, the case in which this too is regular is called an equitable partition (EP). The lack of constraints within partitions cells of an AEP allow for them to exist at multiple, overlapping, scales. 

One can generalize this notion and define an AEP in the context of a weighted graph.\cite{abiad_characterizing_2022} Let $G=(V,E,w)$ be a weighted graph with weight function $w:E\to\R^+$. Then a weighted almost equitable partition of $G$ is a vertex partition $\pi=\{V_1,\dots,V_k\}$ such that for $V_i, V_j\in \pi$ with $i\neq j$, there exists $d_{ij}\in \R^+$ such that for all $v\in V_i$ we have $\sum_{w\in N(v) \cap V_j} w(\{v,w\}) =d_{ij}$. In other words, a weighted AEP is again a vertex partition with regular inter cell connectivity, except now the regularity comes in the form of a constant sum of edge weights rather than a constant number of neighbors. Note that we recover the original definition of an AEP with a weight $1$ uniform graph. For the rest of this paper our discussions on AEPs refer to the weighted AEP unless otherwise specified.

We can use this notion of vertex partitioning, or vertex role assignment, to coarse-grain the graph, summarizing the implied structure in what is known as the quotient graph (See figure \ref{fig:EXPOSITION} C.1).
Consider a graph $G$ with arbitrary partitioning $\pi$. Define the indicator matrix of $\pi$ as $P^\pi\in\R^{n\times k}$ such that $P^\pi_{il}=1$ if vertex $i$ is in partition cell $V_l$, and $0$ otherwise. Note we omit the $\pi$ superscript when the choice of partition is clear. The rows of $P$ correspond to the vertices of $G$ and its columns correspond to the partition cells of $\pi$. Now given an $n$ by $n$ matrix representation of a graph like the graph Laplacian (or adjacency), $L$, we define the quotient Laplacian (or adjacency) matrix of $L$ with respect to the partition $\pi$ as 
\[
L^\pi= (P^TP)^{-1}P^TLP.
\]
This definition uses $\pi$ to partition $L$ into blocks, whose average row sums become the entries of $L^\pi\in \R^{k\times k}$. This means the entry $(L^\pi)_{pq}$ for $p\neq q$ is the negative average edge weight sum from vertices of $V_p$ into $V_q$. The $p$th diagonal entry of $L^\pi$ then corresponds to the sum of the average out weight sums from $V_p$ to all other partition cells, ensuring that the rows of $L^\pi$ sum to zero. As a graph is entirely described by its graph Laplacian, this $L^\pi$ induces it's own graph which we call the quotient graph of $G$.

\begin{figure}
    \centering
    \includegraphics[width=\linewidth]{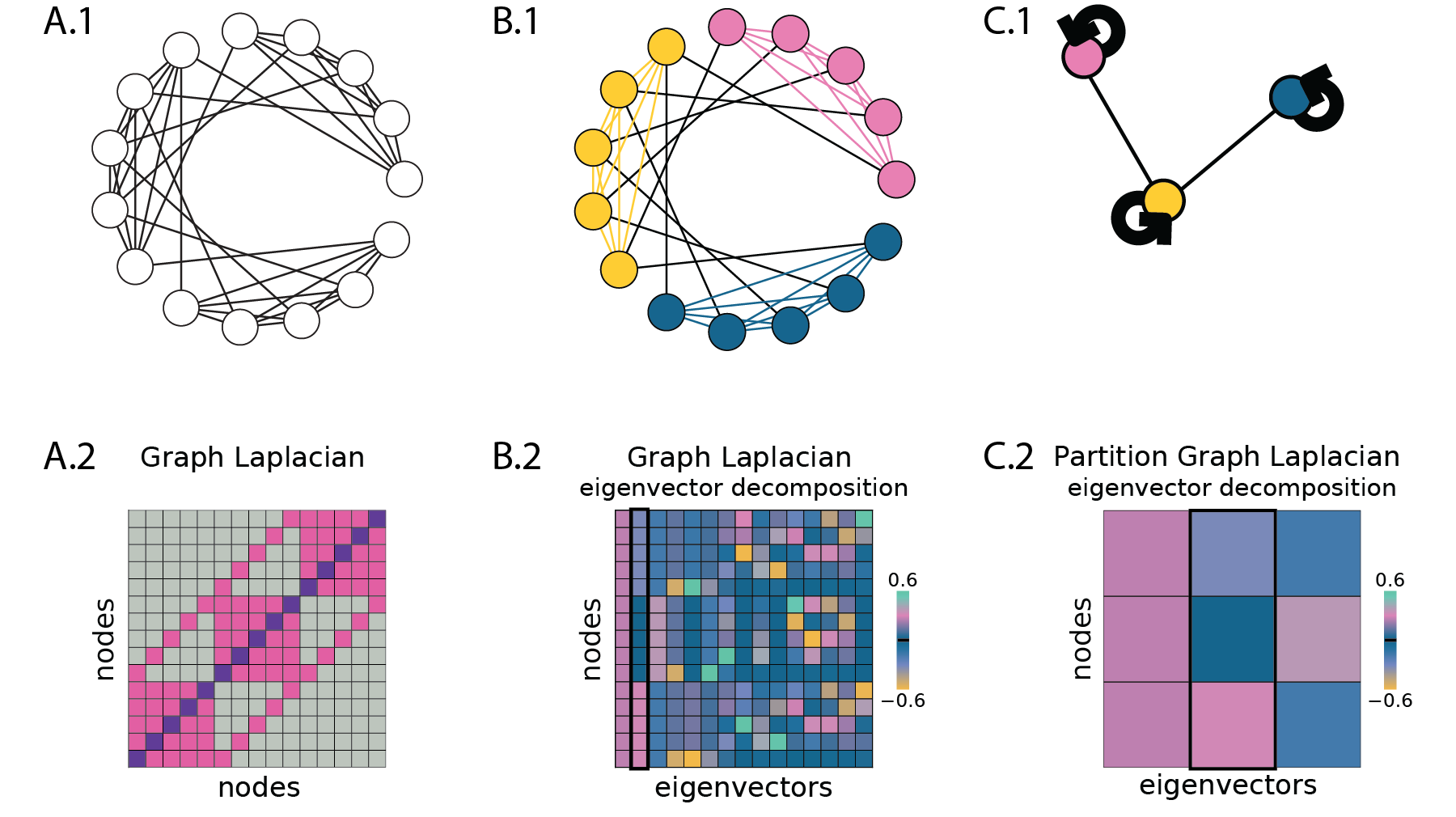}
    \caption{
    In \textbf{A.1}, we show the original graph, and in \textbf{A.2}, its associated graph Laplacian matrix. The graph admits an almost equitable partition (AEP) into three cells, 
    visually suggested in \textbf{B.1} by vertex coloring.
    In \textbf{C.1}, we display the quotient graph induced by the AEP,
    and in \textbf{C.2}, the eigenvectors of its quotient Laplacian.
    The structure of this quotient graph is reflected in eigenspectrum of the graph Laplacian \textbf{B.2}, with the highlighted eigenvector revealing the partition through its block structure. In particular, this structural eigenvector is the scaled counterpart of the highlighted eigenvector in \textbf{C.2}. This highlights the alignment between the structural features of the full graph and its coarse representation.
    }
    \label{fig:EXPOSITION}
\end{figure}

The quotient graph provides a powerful tool for coarsely analyzing the structure of the original graph, while preserving connectivity information relevant to the partition. Indeed, cluster synchronization can be understood as a reduction of the dynamics onto such a quotient graph.\cite{schaub_graph_2016} The preservation of information is made clear in the analysis of the quotient Laplacian in the context of an almost equitable partition. Notice, in the case of an almost equitable partition $\pi$, the off diagonal blocks of the induced block partition of $L$ have average row sums exactly equal to one another. This implies that the off diagonal entries of $L^\pi$, which correspond to the average of these row sums, preserve all the information of inter cell connectivity. The following theorem states this algebraically.

\begin{theorem}\label{thm:quotientAEP}\cite{brouwer2012spectra, cardoso_Laplacian_2007, aguilar_almost_2016} Let $L$ be the Laplacian of a graph $G=(V,E,w)$, and $\pi = \{C_1,\cdots,C_k\}$ be a partition of $V$.\\ The partition $\pi$ is an almost equitable partition \textbf{if and only if} the image of $P^\pi$ is invariant under $L$, which means there exists a matrix M such that $L P = P M$. Moreover, the solution to this Sylvester matrix equation is unique, namely $M = L^{\pi} = (P^TP)^{-1}P^TLP$.
\end{theorem}
Theorem \ref{thm:quotientAEP} comes from a more general theory of the tight interlacing the eigenvalues of $L^\pi$ have with respect to the eigenvalues of $L$.\cite{brouwer2012spectra} 
 It is helpful to note that the expression for the quotient Laplacian $L^{\pi}$ can be written using the Moore-Penrose pseudoinverse as $L^{\pi} = P^{+} L P$. This follows from Penrose's 1955 classical paper\cite{penrose1955generalized} and offers an alternative derivation of Theorem 1. 
Then given, $(\lambda,v )$, an eigen-pair of $L^\pi$, it follows
\[
L^\pi v=\lambda v \implies PL^\pi v = \lambda P v \implies LPv = \lambda Pv
\]
The reverse implication also holds because the columns of $P$ are independent. This leads to the following corollary.

\begin{cor}\label{cor:quotientevec}\cite{brouwer2012spectra,aguilar_almost_2016} 
    Let $\pi$ be an almost equitable partition of a graph, and $P$ the indicator matrix associated to $\pi$. Then $(v,\lambda)$ is an eigenvector-value pair of $L^\pi$ if and only if $(Pv,\lambda)$ is an eigenvector-value pair of $L$.
\end{cor}

This corollary tells us how the spectrum of the quotient graph finds its way into that of the full graph. This relationship is exemplified in figure \ref{fig:EXPOSITION}. The same is true for all AEPs and their corresponding quotient graphs, admitting any structure within partition cells and not just the fully connected one shown in this toy example. We refer to these eigenvectors from the quotient graph as the structural eigenvectors. We borrow this naming convention from network science, the relation to which is discussed in section \ref{discuss}.

While the above theorems and corollaries are defined on graph laplacians, they have counterparts with adjacency matrices. The associated partitions are equitable partitions (EPs) rather than AEPs however. Generally, AEPs are associated to graph laplacians whereas EPs are associated to adjacency matrices.

\subsubsection{Eigenvectors as a Basis for Dynamics}\label{sec:Evecbais}

From PDEs and solutions to the Helmholtz equation, we find intuition for understanding the dynamics on a graph through the eigenvectors of the graph Laplacian. In the solution to the heat equation or wave equation on a surface, the corresponding eigenfunctions represent 'spatial frequencies' that are composed in a Fourier series to describe solution functions on the surface.\cite{haberman2013applied} The graph Laplacian can be thought of as a discretized Laplacian operator, abstracted to a graph. This spatial structure is both discrete and in general non-uniform. 

The Laplacian matrix emulates diffusion across the graph, analogously to the Laplace operator.\cite{belkin2008towards} The matrix has (together with its deformed variants) a well established connection with the diffusive process of random walks on graphs.\cite{Seabrooketal_Markov_Spectral,boukrab_random-walk_2021} 
Further, its eigenvectors and eigenvalues encode 'spatial frequency', playing a similar role to the eigenfunctions of the Helmholtz equation in the case of diffusion on a surface, but discretized over a non-uniform graph structure. One way in which we can see this reflected is by measuring zero crossings along random walks over the graph.\cite{shuman_emerging_2013} The rows and columns of the graph Laplacian are associated to the vertices of the graph. In analogy to functions, the eigenvectors of this matrix are of length $n$ and assign a value to each vertex. In considering random walks on our graph, one can count the frequency of sign changes along those walks and see positive correlation with the eigenvalue of that eigenvector.\cite{shuman_emerging_2013} This yields an intuition of 'spatial frequency' and how it is abstracted to graphs. 
A relevant example is the Fiedler eigenvector, associated to the first non-zero eigenvalue of the graph Laplacian. The vertex assignement here can be thought of as a gradient from negative to positive values across the graph, which naturally presents a low 'spatial frequency'.
This idea opens the door for us to construct a graph-equivalent of the Fourier series solution method for PDEs, describing solutions in terms of 'spatial frequencies'.\cite{shuman_emerging_2013}

Due to the symmetry of the graph Laplacian, its eigenvectors are real, orthogonal, and define a basis of $\R^n$. Thus the eigenvectors can be used as a basis to describe behavior on the vertices of a graph. We may then ask if structural eigenvectors can form a sub-basis, since they come from a Laplacian as well, allowing for a restricted description of behavior on the associated partition. In general a quotient Laplacian will not be symmetric, preventing us from using the properties of symmetric matrices as before. While the structural eigenvectors of $L$ are orthogonal, the corresponding eigenvectors from $L^\pi$ are not necessarily so. We can show that they form a basis however.
\begin{prop}\label{prop:quotientspan}
    Let $G$ be a graph, and $L$ its graph Laplacian. Let $\pi$ be an AEP of size $k$ on $G$, and $L^\pi$ be the quotient Laplacian of $L$ associated to $\pi$. The eigenvectors of $L^\pi$ span $\R^k$.
\end{prop}
\begin{proof}
    Let $L^\pi$ be the quotient Laplacian as described above. It follows
    \begin{align*}
        (L^\pi)^T L^\pi &= (P^TP)^{-2}\left(P^T L^T P\right)\left(P^T L P\right)\\
        &= (P^TP)^{-2}\left(P^T L P\right)\left(P^T L^T P\right)    &&\text{since $L$ is symmetric}\\
        &=L^\pi(L^\pi)^T .
    \end{align*}
    Thus $L^\pi$ is a normal matrix. It follows that it is diagonalizable and admits a full dimensional eigenspace.
\end{proof}

Furthermore, the structural eigenvectors in the spectrum of a graph Laplacian, as shown in Corollary \ref{cor:quotientevec}, are within the column space of $P$, and therefore assign the same values to vertices within each partition cell. Thus the structural eigenvectors define a sub-basis in the spectrum of the graph Laplacian that captures all information relevant to the clustered dynamics, allowing the system to be effectively described by the quotient graph. We can formalize this conclusion in the following statement.

\begin{cor}\label{Cor:SubBasis}
    Let $G$ be a graph, $L$ its graph Laplacian, and $\pi$ an AEP on its vertex set. The set $\{P^\pi v | v \text{ eigenvector of $L^\pi$}\}$ forms a sub-basis of the eigenvector basis of $L$ which spans the space of vectors that are constant on the cells of $\pi$, and therefore encodes all dynamics constrained by the partition.
\end{cor}
\begin{proof}
    Given the graph $G$ and an AEP $\pi$ of size $k$ on its vertex set, corollary \ref{cor:quotientevec} guarantees that $B = \{P^\pi v | v \text{ eigenvector of $L^\pi$}\}$ is a subset of eigenvectors of the graph Laplacian. By proposition \ref{prop:quotientspan}, the set of eigenvectors of $L^\pi$ span $\R^k$. It follows that $B$ spans the column space of $P^\pi$. Recall that vectors in this column space are of size $n$ and have constant values for all entries corresponding to vertices within the same cell of $\pi$. Thus $B$ is a subspace of the eigenspace of $L$ which spans all vector valued vertex assignments constant within the cells of $\pi$ in $G$.
\end{proof}

In section \ref{K-Dynamics} we will see how the eigenvalues become relevant to the synchronization of dissipative models like the Kuramoto model. We draw from the notion of assortativity in community structure to understand the placement of structural eigenvectors in the spectrum of the Laplacian. An assortative partitioning of the vertex set has greater intra cell than inter cell connectivity. A disassortative partitioning has the converse. Moreover, the assortativity of the partitioning has been linked to the placement of the structural eigenvector in the spectrum of the Laplacian.\cite{krzakala_spectral_2013} In particular, for large enough $n$, structural eigenvectors associated to assortative partitions tend towards smaller associated eigenvalues, whereas the structural eigenvectors associated to disassortative partitions tend towards the larger eigenvalues. The theory for this comes from community detection,\cite{saade_spectral_nodate} but we can explain this phenomenon heuristically through the lens of AEP vertex assignments and random walks with structural eigenvectors. Given a highly assortative AEP, random walks are likely to remain inside the partition cell they start from, and since the vertex assignments of the corresponding eigenvector are constant within those cells, the random walk is likely to correspond to relatively few zero crossings in a corresponding structural eigenvector. This lower number of zero crossings for the associated structural eigenvectors suggests a lower eigenvalue. A similar argument follows for a disassortative AEP. 

\subsection{Kuramoto Dynamics}\label{K-Dynamics}

The Kuramoto model represents vertices as oscillators and  couples their differential equations over the edges of the network. 
We consider a Kuramoto model on a connected graph, as in equation \ref{Kuramoto}, and represent the oscillator phases in vector form as $\vec{\theta}(t)\in \R^n$.
Following section \ref{sec:Evecbais}, we can decompose $\vec{\theta}$ using the eigenvector basis of $G$, where $v_i$ is the $i$th normalized eigenvector of the graph Laplacian, with an ordering imposed by the accompanying eigenvalues:
\[
\vec{\theta}(t) = \sum_{i=1}^n \alpha_i(t) v_i
\]
Previous work has shown that with some persuasion, one can perform a change of variable of equation \ref{Kuramoto} to isolate the coefficients of this decomposition.\cite{kalloniatis_incoherence_2010} We extend this work to include weighted networks and antisymmetric phase offsets. In this section, we only consider the weighted extension, seen below, and we study the latter in section \ref{sec:phase lag},

\begin{align}\label{coeffeqn}
    \dot{\alpha}_r &= \omega^{(r)} - \sigma \sum_{a\in E(G)} W_{aa}   e_a^{(r)} \sin\left( \sum_{s>0} e_a^{(s)}\alpha_s(t)\right)
    && r\not=0\\
    \alpha_0 &= \sqrt{N} \overline{\omega} t.
\end{align}

Here $\omega^{(r)}=\vec{\omega}\cdot v_r$, is the product of the vector of natural frequencies with the $r$th eigenvector and $W_{aa}$ is the weight of edge $a$. Because our graph has a single connected component, the eigenvector corresponding to the zeroth eigenvalue is a vector of constant values. Thus isolating this normal mode is equivalent to quotienting out by the average frequency. For the zeroth coefficient, $\overline{\omega}$ is the average of the natural frequencies $\vec{\omega}$. The scaling of $\sqrt{N}$ comes from the normalization of the zero mode eigenvector. These equations then take this a step further by describing the rest of the model in the eigenbasis.

We call this basis change of the Kuramoto equations, the coefficient equations. This transformation frames the problem of synchronization in terms of the population dynamics of coefficients. When the system reaches a globally locked frequency, then naturally the zeroth mode describes the advancement of all phase angles, while all other coefficients describe the phase offsets of the system. Thus in this form of synchronized behavior,  the coefficients for non-zero eigenvalues approach constant values.\cite{kalloniatis_incoherence_2010} Assuming coefficients are small and close to synchronization, we can solve a linearization of equation \ref{coeffeqn}. For $r\neq 0$:
\begin{equation}\label{eqn:linearcoeff}
    \alpha_r(t) = \frac{\omega^{(r)}}{\sigma \lambda_r}(1-e^{-\sigma\lambda_r t}) + \alpha_r(0)e^{-\sigma \lambda_r t}.
\end{equation}
This allows us to find asymptotic limits for the coefficients during synchronization, $\alpha_r^\infty$. These coefficients describe the phase offsets once synchronized. We find that in a synchronizing regime on a weighted graph, the asymptotic behavior of the $r$th coefficient approaches
\begin{equation}\label{eqn:limitcoeff}
    \alpha_r^{\infty} =\frac{\omega^{(r)}}{\lambda_r \sigma}.
\end{equation}
We note that the form of this equation does not require the weights of the graph. All structural information related to edge weights, impacting the synchronization in the system, is encoded in the eigenvalues and eigenvectors of the graph Laplacian. Thus this limit is significant in that it reflects the structure present in the system.

\subsection{Spectral Kuramoto Clustering}

Previous work has shown that a weighted AEP can induce cluster synchronized behavior in the Kuramoto model.\cite{kato_cluster_2023,menara_stability_2020,schaub_graph_2016} We use the spectral properties of AEPs to recast this behavior in the context of the coefficient equation. This framework provides us with inroads to explore more general synchronized behavior.

Using the coefficient end behavior in equation \ref{eqn:limitcoeff}, we can describe the impact of partitions and graph structure on synchronization.
If the natural frequency vector has constant values within the cells of an AEP, then by corollary \ref{Cor:SubBasis}, $\vec{\omega}$ is described completely within the scaled up eigenbasis of the corresponding quotient Laplacian. By orthogonality, all eigenvectors not from the quotient Laplacian must be orthogonal to $\vec{\omega}$. Thus the linearized coefficient limit is only nonzero for the structural eigenvectors. As a result, only the structural eigenvectors are nonzero in the end-state eigendecomposition of the system (figure \ref{fig:2 AEP Phase decomp}). Expressed in terms of the dynamics of the original Kuramoto system, this suggests that the system experiences cluster synchronization, where the phases of oscillators within the clusters of the AEP remain equal. The significance of this choice is that the natural frequency distribution over the vertices excites the structural modes rather than the nonstructural ones, allowing for these exact clustered dynamics. This aligns with work such as\cite{fan_eigenvector-based_2023} which has found that synchronized clusters in similar models of chaotic oscillators require identical elements in some eigenvector (figure \ref{fig:2 AEP Phase decomp}).

    We note that since equation \ref{eqn:limitcoeff} is from a linearization of the system, the approximation it gives diverges for large end state values of coefficients. For fully synchronized systems, the range of error is proportional to the coefficient in question (figure \ref{fig:Linear Div}). Thus the modes orthogonal to the natural frequency vector respect the approximation, and the qualitative dynamics as a whole remain the same.
    
    \begin{figure}
        \centering
        \includegraphics[width=\linewidth]{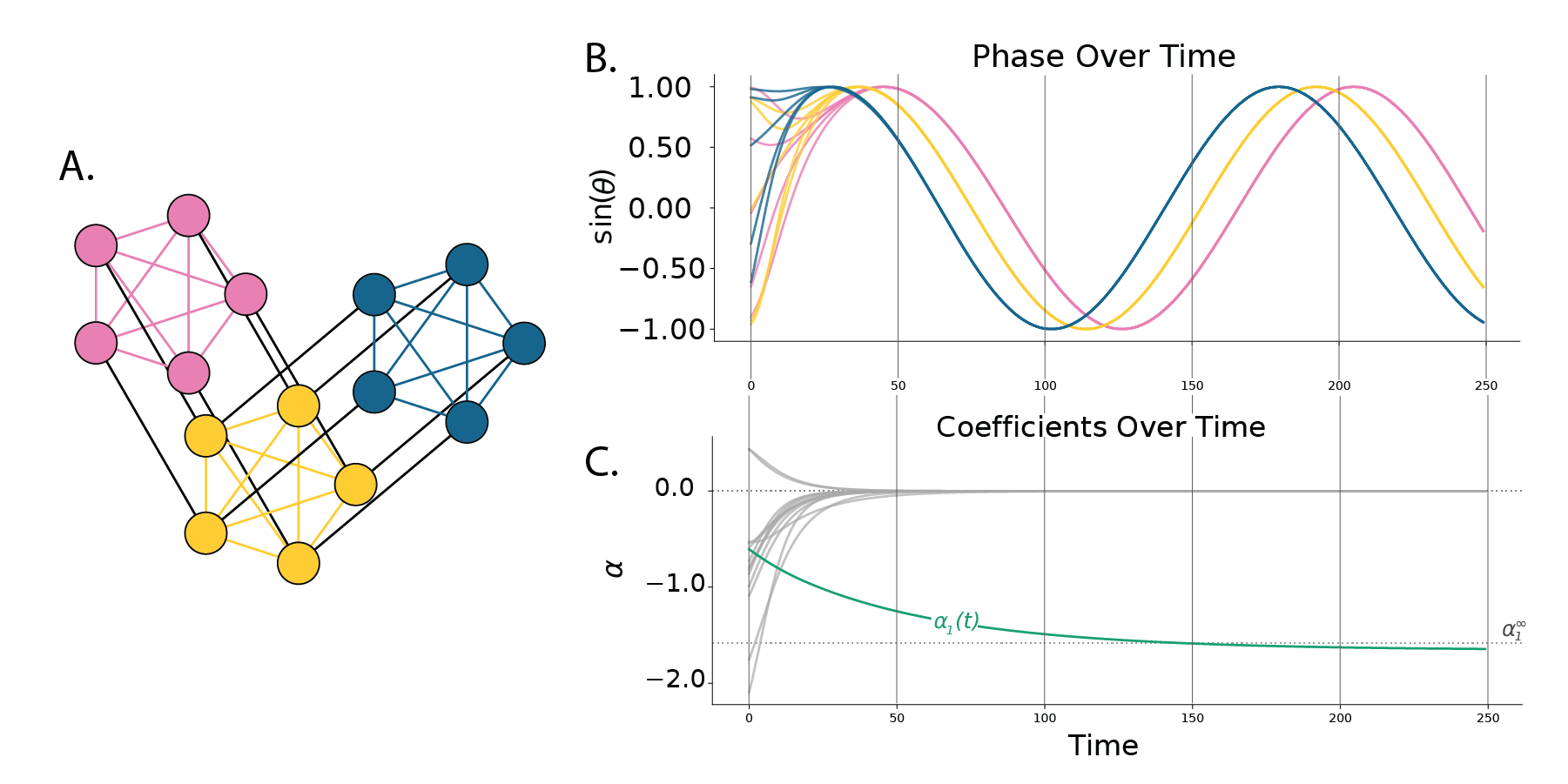}
        \caption{This example highlights the impact of structural eigenvectors on clustered synchronization. \textbf{A. The network} consists of 15 vertices, color-coded according to an almost equitable partition (AEP) of the network. A Kuramoto model is placed on the graph, with coupling $\sigma=2$ and natural frequencies set to be identical within each cell of the AEP and distinct across cells (respectively the pink, yellow, blue cells have frequencies 3,4,5). This frequency assignment reinforces the structural partition. \textbf{B. The resulting phase dynamics} show convergence to a phase-clustered regime consistent with the AEP. \textbf{C. Time evolution of the projection coefficients onto the Laplacian eigenbasis.} Cluster synchronization corresponds to the dominance of the structural subspace, here captured primarily by $\alpha_1$, the eigenmode associated with $\lambda_1$, whose asymptotic limit is approximated by $\alpha_1^\infty$.}
    \label{fig:2 AEP Phase decomp}
\end{figure}

We may also notice that the end state coefficients depend on the eigenvalues as well. This recovers the intuitive notion that more assortatively partitioned networks will synchronize more readily in the Kuramoto model. We saw before that assortative structural eigenvectors generally correspond to lower eigenvalues, and here we observe that the lower eigenvalues have larger end coefficients. This means that when keeping the alignment to the natural frequency vector the same, assortative partitions which have lower corresponding eigenvalues support larger differences in phase between clusters. Or, equivalently, they support cluster synchronization at a lower coupling strength, $\sigma$. We can see this phenomenon supported in figure \ref{fig:big_net_sync}.

\begin{figure}
    \centering
    \includegraphics[width=0.7\linewidth]{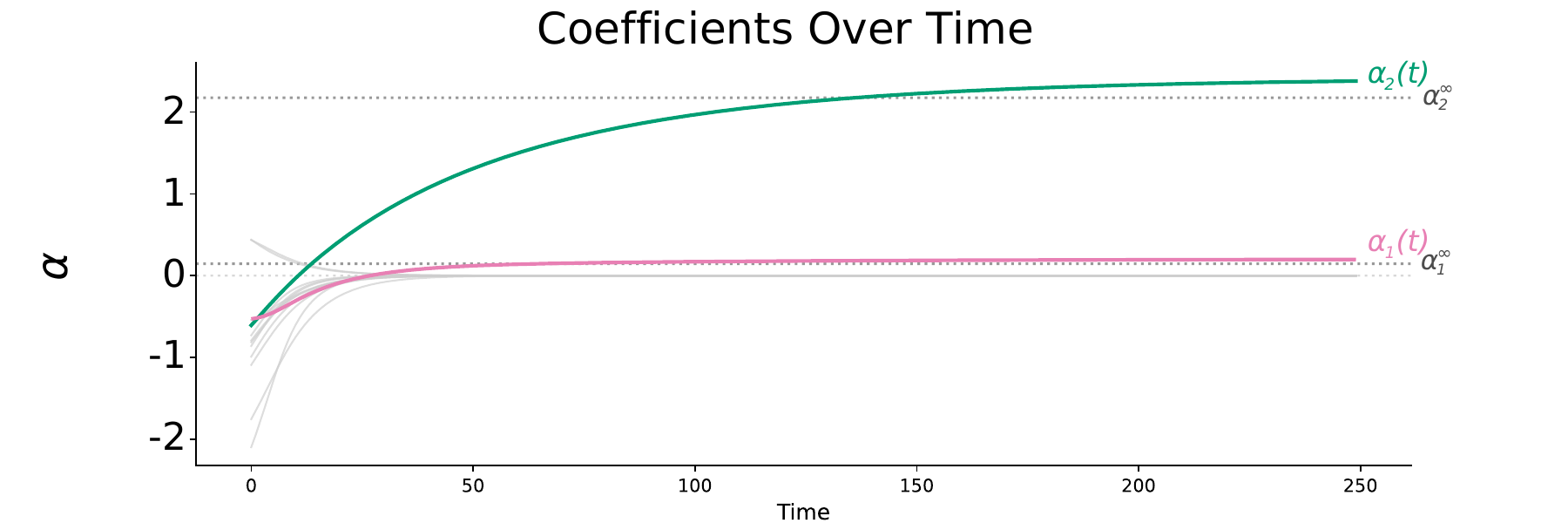}
    \caption{An example to highlight the convergence behavior of spectral coefficients using the toy network presented in figure \ref{fig:2 AEP Phase decomp}.A with the same AEP on the vertices but different parameters in the Kuramoto model. In this example, we set the natural frequencies to be identical within each cell of the AEP and distinct across cells (respectively the pink, yellow, blue cells have frequencies 3,4,9), and we increased the coupling strength between oscillators to 5.
    which leads to the coeffients $\alpha$ having a slightly different behavior than the ones in figure \ref{fig:2 AEP Phase decomp}.C. We see that the deviation between the time-dependent coefficients and their linearized limit values decreases with the magnitude of the corresponding mode. In particular, $\alpha_1(t)$ converges more closely to $\alpha_1^\infty$ than the higher-index mode $\alpha_2(t)$, exemplifying this behavior.}
    \label{fig:Linear Div}
\end{figure}

\subsection{Transient Cluster Synchronization}\label{sec:transient}

We can use the same framework to understand full phase synchronization as well. 
A necessary condition for full phase synchronization in the Kuramoto model is uniform natural frequencies. Thus we assume $\omega_i=\delta$ for all vertices $i$. It follows $\vec{\omega}=\delta \mathbb{1}$. Recall that the smallest eigenvalue is $\lambda_0=0$ and the corresponding eigenvector $v_0$ is the constant vector. Thus in the eigenbasis of the graph Laplacian, $\vec{\omega}$ is entirely described by $v_0$. It follows that $\omega^{(r)}=0$ for all $r\neq 0$, and therefore $\alpha^\infty_r=0$ for all $r\neq 0$. Thus the limit for the dynamics for large $t$ is
\[
\vec{\theta}(t) = \alpha_0(t) v_0 = \sqrt{N} \delta t \mathbb{1}.
\]
This describes a fully synchronized system of oscillators, recovering the end state of a fully phase synchronized system.

Using the equation \ref{eqn:linearcoeff}, we can see the path the eigenmodes take to equilibrium. We observe that in the described setting, for $r\neq 0$ the solution follows $\alpha_r(t)=\alpha_r(0)e^{-\sigma\lambda_r t}$. This indicates faster convergence for the larger eigenvalue modes. In other words, with comparable initial values, a difference in eigenvalues leads to the separation of convergence rates to synchronization. This explains the observations of previous works,\cite{de_domenico_diffusion_2017, arenas_synchronization_2006} which found convergence to synchrony as being in clusters, with delays in transient clustering regimes being proportional to the gaps in the eigenspectrum of the Laplacian. If the network has non trivial AEPs, those would show up in the eigenspectrum and if their placement in the graph Laplacian spectrum is low, then their coefficients will be the slowest to approach zero (figure \ref{fig:big_net_sync} C.). In the presence of an AEP, we define cluster synchronization as a restriction of the dynamics to the structural sub-basis. Using corollary \ref{Cor:SubBasis} this means that the described scenario results in transient clustered behavior while approaching equilibrium (figure \ref{fig:big_net_sync} B). The key here is that gaps in the eigenspectrum can relate to regime changes when approaching synchronization because the rate of asymptotic decay to the limit of a coefficient is determined by the eigenvalue of that mode. It is summarized in proposition \ref{prop:transient}, with the proof in appendix \ref{apdx:transient}.

\begin{prop}\label{prop:transient}
     Let $G$ be a connected graph with an AEP, $\pi$. Suppose $\pi$ is sufficiently assortative in that the $p\geq 1$ smallest eigenvalues of the graph Laplacian, $L$, are associated to corresponding structural eigenvectors. Assume there are no repeated eigenvalues. Consider a Kuramoto model on $G$ with uniform natural frequencies, and assume $\alpha_i(0)>0$ for some $0<i\leq p$, Then in the linearized solution the system will exhibit transient cluster synchronization along $\pi$.
\end{prop}

As usual the above clustered behavior is defined as the domination of the structural sub-basis. However, it is important to note the impact of the choice of initial condition. All initial conditions satisfying \ref{prop:transient} will eventually be dominated by the assortative structural eigenmodes, but choices in which the structural modes have small initial condition will result in all nonzero modes being small by the time the structural modes dominate. In the phase dynamics this is reflected as clustered dynamics where the clusters are very close in phase to one another. 

Since AEPs can occur at multiple scales, this can also describe transitions through multiple clustering regimes.
Consider a graph $G$ with a sequence of possible almost equitable partitions $\pi_1, \pi_2, \dots, \pi_k$ such that $\pi_1 \prec \pi_2 \prec \dots \prec \pi_k$, meaning each partition is a refinement of the next (i.e., each cell in $\pi_i$ is fully contained in a cell in $\pi_{i+1}$). We say this sequence defines a $k$-layer hierarchical AEP clustering of $G$.\\

In a graph with assortative hierarchical clustering, i.e. each $\pi_i$ is assortative as defined in section \ref{sec:Evecbais}, then the structural eigenvectors of coarser AEPs generically have lower eigenvalue than the finer ones.\cite{krzakala_spectral_2013, schaub_hierarchical_2023} If assortative enough, these structural eigenvectors occupy the lower frequency modes (see an example of such a network in figure \ref{fig:big_net_sync} A). While approaching full synchronization this would lead to higher frequency, non structural, modes dying initially to reveal the more slowly decaying structural modes, whose overlap would manifest as a finely clustered behavior (see lemma \ref{lem:refinement} in appendix \ref{apdx:transient}). Then in time the lower frequency structural modes dominate the higher frequency structural modes as well. Assuming similar starting values, the duration of these separate clustering regimes is proportional to the gap in the eigenspectrum between sequential modes. The following result naturally follows from the previous, with the full proof in appendix \ref{apdx:transient} and example in figure \ref{fig:big_net_sync}.

\begin{prop}\label{prop:transienthier}
    Let $G$ be a graph with $k$-layer hierarchical AEP clustering defined by a sequence of nested partitions $\pi_1 \prec \pi_2 \prec \dots \prec \pi_k$. Assume assortativity such that for some choice of indices $0=p_0<p_1<\cdots<p_k<n$, for each $0<i\leq k$ the eigenvalues $\lambda_{p_{i-1}+1}<\dotsc < \lambda_{p_i}$ correspond to structural eigenvectors of $\pi_i$. Assume all eigenvalues are simple (non-repeated).
    Consider a Kuramoto model on $G$ with uniform natural frequencies. Then up to the choice of initial condition, in the linearized regime the system will pass through a sequence of transient cluster synchronization states corresponding to each level of the AEP hierarchy, before ultimately converging to global synchrony.
\end{prop}

\begin{figure}
    \centering
    \includegraphics[width=\linewidth]{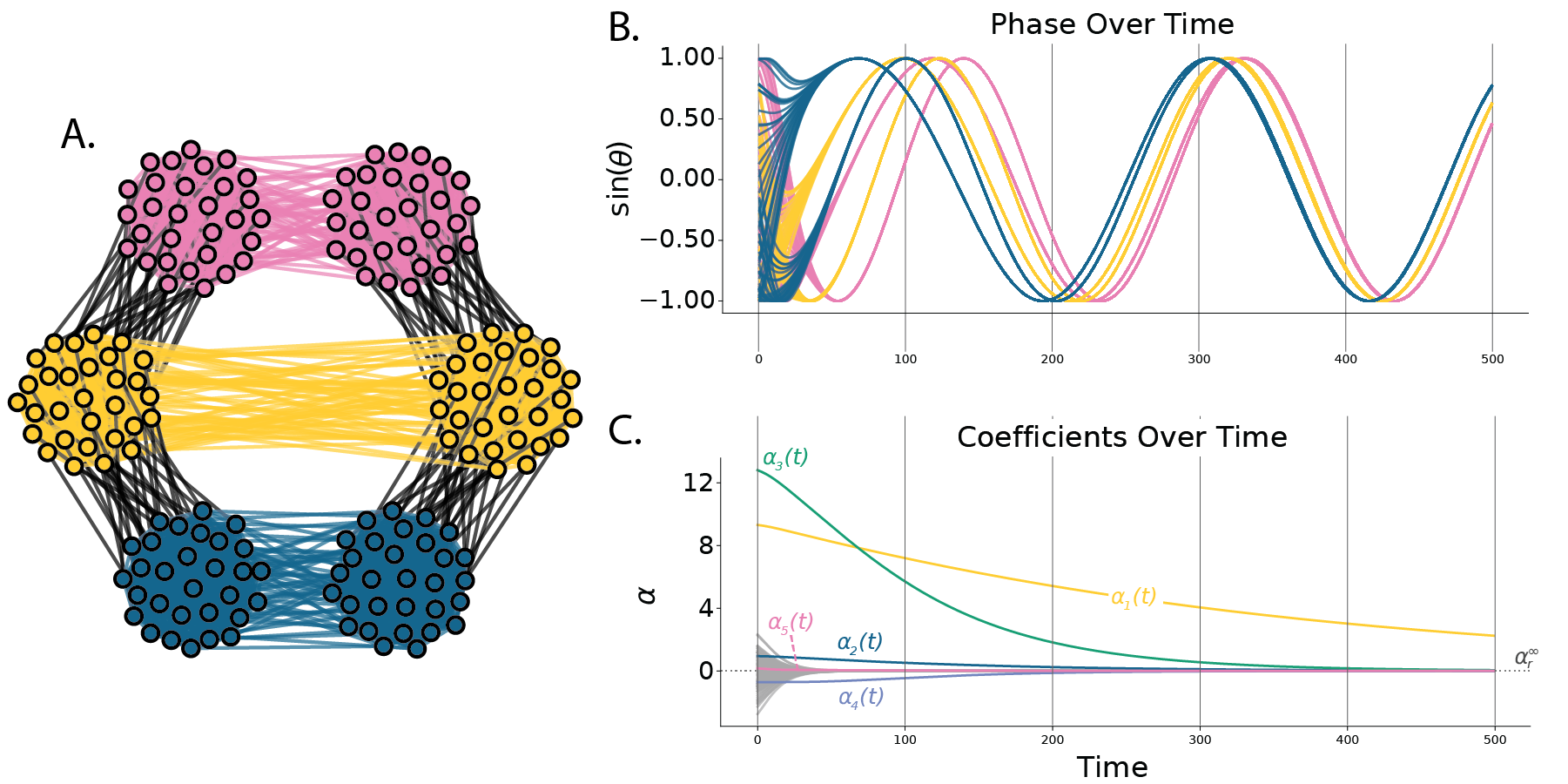}
    \caption{Hierarchically related structural modes. \textbf{A. A hierarchically clustered network with 180 vertices.} Vertex positions illustrate one of the allowable AEPs with 6 clusters, while vertex and edge coloring highlight a coarser AEP consisting of 3 clusters. To ensure full system synchronization, the natural frequencies were uniformly set to 3 across all vertices along with a coupling of $\sigma=0.3$. \textbf{B. Phase dynamics showing convergence to full synchronization.} Four distinct dynamical stages are observed: an initial disordered regime (time steps 0–50), followed by the formation of 6 clusters (time steps 50–350); these then merge into 3 larger clusters, which gradually converge into a single synchronized group. \textbf{C. Time evolution of projection coefficients onto the Laplacian eigenbasis}, illustrating the same sequence of transitions. Initially, a wide distribution of coefficients reflects disorder; this is followed by the decay of nonstructural modes ($\alpha_r$, $r>5$), leaving only the structural components ($\alpha_r$, $r\leq5$). Subsequently, the finer-scale structural modes decay, and the dynamics are eventually dominated by an eigenmode corresponding to the coarsest AEP, $\alpha_1$. The ultimate convergence to global phase synchronization is reflected by all coefficients asymptotically approaching their limits $\alpha_r^\infty=0$ for all $r$, a result enabled by the uniform assignment of natural frequencies across all vertices.}
    \label{fig:big_net_sync}
\end{figure}

\section{Stability of Structurally motivated Clustering}

In section \ref{sec:spectral_description_of_synchronized_clusters} we have shown how with the coefficient equation, we can frame cluster synchronization in terms of the dynamics of AEP structural eigenvectors of the graph Laplacian. The natural question is then what other kinds of clustering the Laplacian basis could allow for. Framing this problem in the linear algebra of an eigenbasis and the AEP provides a natural direction for loosening the structural restrictions for synchronized clusters. Strict AEPs are rare and the associated full phase synchronization within clusters may be stronger than what is required for application. The concept of a weakened AEP has been considered previously, such as in the model of stochastically equitable partitions,\cite{schaub_hierarchical_2023} and even in the setting of the Kuramoto model with a bound on the edge weight sums between partitions.\cite{kato_cluster_2023} We approach this question from the perspective of the spectral theory behind an AEP to see how sensitive the eigenbasis is to perturbation, and how a spectral theory around its weakening could fit into the spectral description of synchronization presented in this paper.

\subsection{Approximate Almost Equitable Partitions}

In theorem \ref{thm:quotientAEP} we saw that given a partition $\pi$ with indicator matrix $P$ on a graph with Laplacian $L$, a sufficient and necessary condition for $\pi$ to be an AEP is that $LP=PL^\pi$. The entry $(LP)_{ip}$ corresponds to the edge weight sum of vertex $i$ to all its neighbors in cluster $p$ of $\pi$. Meanwhile the matrix $PL^\pi$ inhabits the same form, but in the $p$th column, assigns to each vertex in a given cluster the average edge weight sum from that cluster to cluster $p$. Thus we define $E = PL^\pi - LP$ which captures the deviation of the vertex edge weight sums from the averages and allows us to approximate the spectral AEP property. The minimization of a similar matrix was explored in\cite{scholkemper_optimization-based_nodate} in the context of finding regular equitable partitions. 

Let us consider what the introduction of $E$ allows us to say about the eigenvectors of $L$. If $v$ is an eigenvector of $L^\pi$ with eigenvalue $\lambda$, then
\begin{align*}
    L^\pi v = \lambda v &\implies PL^\pi v = \lambda P v\\
    &\implies (LP+E)v=\lambda Pv\\
    &\implies L(Pv) = \lambda(Pv) - Ev\\
    &\implies Lw = \lambda w - \epsilon.
\end{align*}
Where $w=Pv$ is the scaled up quotient Laplacian eigenvector and $\epsilon = Ev$ is an error term expressing how the transformation of $v$ by $E$ affects the final equation. We call this $\epsilon$ the equitable error and are concerned with cases in which this error is small, indicating a minor deviation from the spectral AEP property. 

Given a graph $G$ and its graph laplacian $L'$, let $L'=L+N$ where $L$ is a graph laplacian of matching dimension which admits an AEP $\pi$. Then $N$ captures the deviation of the true Laplacian $L'$ from the almost equitable $L$. It follows 
\begin{align*}
    E &= L'P - P (L')^\pi  \\
    &= (L+N)P - P(P^T P)^{-1}P^T(L+N)\\
    &=LP - P(P^T P)^{-1}P^TL + NP - P(P^T P)^{-1}P^T N\\
    &=NP - P(P^T P)^{-1}P^TN\\
    &=NP - P N^\pi.
\end{align*}
When put in this form, the Laplacian L, which supports an AEP cancels out. This shows how $E$ can be described in terms of the deviation from a graph Laplacian supporting the AEP. Thus it is the placement of the noise over the graph that determines the equitable error. This suggests that the probability distribution over which error is introduced could have profound impact on how it is represented in the equitable error (Appendix \ref{app:SBMlimit} explores one such avenue). We can bound the norm of $\epsilon$ and relate the values of $E$ to the size of $\epsilon$ more explicitly.

\begin{prop}
    Let $G$ be a weighted graph with graph Laplacian $L$. Let $\pi$ be a vertex partition. The equitable error of any eigenvector $v$ of $L^\pi$ is bounded as $\|\epsilon\| \leq \sigma_1(E)\|v\|\leq 2k\|v\| \max_i  \sum_{j}\left|e_{ij}\right|$ where $\sigma_1(E)$ is the largest singular value of $E$ and $k$ is the number of cells in $\pi$.
\end{prop}

\begin{proof}
    Recall that for a rectangular matrix $A$ the largest singular value, $\sigma_1(A)$, bounds its products, $\|Av\|\leq \sigma_1(A) \|v\|$. It can be shown that $\max_{x\not=0}\frac{\|Ax\|_\infty}{\|x\|_\infty} = \max_i \sum_{j} |a_{ij}|$. The largest singular value of a matrix is equal to its operator norm. Thus it follows that 
    
    \begin{align*}
        \sigma_1(E)&= \max_{x\not=0}\frac{\|Ex\|}{\|x\|}      &&\text{largest $\sigma_i$ equal to operator norm}\\
        &\leq \max_{x\not=0}\frac{\sqrt{k}\|Ex\|_\infty}{(1/\sqrt{k})\|x\|_\infty}\\
        &= k \max_{x\not=0}\frac{\|Ex\|_\infty}{\|x\|_\infty}\\
        &=k \max_i \sum_{j} |e_{ij}|\\
    \end{align*}
    
    Thus if $v$ is an eigenvector of $L^\pi$, it follows that
    
    \[
    \|\epsilon\| \leq \sigma_1(E)\|v\|\leq 2k\|v\| \max_i  \sum_{j}\left|e_{ij}\right|.
    \]
\end{proof}

The size of this bound is determined by the vertex which differs the most in magnitude from the average out edge weight sums of its cluster. It is also of note that this bound scales with the number of clusters in the partition, $k$. Thus fewer partitions de-emphasize the maximum deviations from the average. This bound represents the absolute worst case scenario, and is not tight. In reality, the vectors $v$ have no reason to be exactly in line with the largest singular vector, so we can expect better results in application.

Bounding the equitable error allows us to indicate how close our structural eigenvectors are to the the scaled up form $Pv$, the form the corresponding structural eigenvectors of an AEP would have.

\begin{prop}\label{prop:Best_Approx}
    Let $G$ be a weighted graph with graph Laplacian $L$. Let $\pi$ be a vertex partition such that the equitable error $\|\epsilon\| \leq \delta$ for eigenvector-value pair $(\lambda,v)$ of $L^\pi$. Then $Pv$ is best approximated by eigenvectors with eigenvalues close to $\lambda$. For any truncated approximation, $u$ using eigenvectors with eigenvalues within $\gamma$ of $\lambda$ and assuming $m$ of these, the bound $\|Pv-u\| \leq (\delta/\gamma) \sqrt{n-m}$ will hold. 
\end{prop}

\begin{proof}
    Consider the basis of eigenvectors of $L$. The Laplacian is symmetric and therefore diagonalizable since the graph is undirected. Therefore its eigenvalues form an orthogonal basis.  Let $\{v_1,\dotsc,v_n\}$ be a set of normal representatives, forming the spectral basis of $L$.

    With the choice of partition $\pi$, we may define the indicator matrix $P$, quotient Laplacian $L^\pi$, and consequently $E=LP-PL^\pi$. Let $(\lambda,v)$ be an eigenvalue eigenvector pair of $L^\pi$. Following from above, we see that $L^\pi v=\lambda v \implies Lw= \lambda w - \epsilon$ where $w=Pv$ and $\epsilon = Ev$.
    
    Consider the expansions $w=\alpha_1 v_1+\alpha_2 v_2+\dotsc+\alpha_n v_n$ and $\epsilon=\beta_1 v_1+\beta_2 v_2+\dotsc+\beta_n v_n$. It follows

    \begin{align*}
        Lw=\lambda w -\epsilon &\implies L(\alpha_1 v_1+\dotsc+\alpha_n v_n) = \lambda(\alpha_1 v_1+\dotsc+\alpha_n v_n)- \epsilon\\
        &\implies 
        \alpha_1 \lambda_1 v_1+\dotsc+\alpha_n \lambda_n v_n = (\lambda \alpha_1 - \beta_1 ) v_1+\dotsc+ (\lambda \alpha_n -\beta_n) v_n\\
        \implies 
        0 &= (\lambda \alpha_1 - \beta_1 - \alpha_1\lambda_1 ) v_1+\dotsc+ (\lambda \alpha_n -\beta_n -\alpha_n\lambda_n) v_n\\
        &= ((\lambda-\lambda_1) \alpha_1 - \beta_1) v_1+\dotsc+ ((\lambda-\lambda_n) \alpha_n -\beta_n) v_n\\
        \implies 
        0 &= ((\lambda-\lambda_i)\alpha_i-\beta_i) \\&\hspace{20pt}\text{for all $i$, by orthogonality}
    \end{align*}
    By the assumption $\|\epsilon\|\leq \delta$ it follows $|\beta_i|\leq\delta$. 
    Suppose that $H$ is the basis transformation matrix for the eigenbasis of $L$ (its rows are the $v_i$). Then $\|\beta\|^2=\langle 
    \beta,\beta \rangle=\langle 
    H\epsilon,H \epsilon \rangle=\langle 
    \epsilon,H^TH\epsilon \rangle=\langle 
    \epsilon,\epsilon \rangle=\|\epsilon\|^2$ since $\langle \cdot,\cdot\rangle$ is a norm and $H$ is orthonormal. Thus $\|\epsilon \|=\|\beta\|$ and therefore any bound on the norm of $\epsilon$ carries over to the norm of $\beta$. Since $\|\cdot\|$ denotes the euclidean norm, it follows $|\beta_i|\leq \|\beta\|$ and therefore that $|\beta_i|\leq \delta$ for $1\leq i\leq n$
    
    Recall that $L$ is positive semidefinite, and therefore has only non negative eigenvalues. From the above result we have that $\beta_i=(\lambda-\lambda_i)\alpha_i$ for $1\leq i\leq n$. Thus $|(\lambda-\lambda_i)\alpha_i| \leq \delta$ for all $i$.
    
    It follows that when $\delta$ is small, $(\lambda-\lambda_i)$ or $\alpha_i$ is small. For a fixed $\beta_i$ value, the further away $\lambda_i$ is from $\lambda$, the smaller $\alpha_i$ must be. This means that eigenvectors of $L$ with eigenvalues close to $\lambda$ can have larger coefficients $\alpha_i$ and therefore contain the most information of $w$. Let $A=\{a_1,\cdots,a_m\}$ be the set of coefficients for eigenvalues within $\gamma$ of $\lambda$. For all other eigenvalues, $\alpha_i$ is likely small. Therefore $u= \alpha_{a_1}v_{a_1}+\cdots+\alpha_{a_m}v_{a_m}$ is our approximation of $Pv$ in the spectral basis. This allows for the derivation of a formal bound.
    \begin{align*}
        \|Pv - u\|^2 &= \|\sum_{i\not\in A} \alpha_i v_i\|^2\\
        &= \sum_{i\not\in A} |\alpha_i|^2    &\text{as shown earlier for basis change}\\
        &\leq \sum_{i\not\in A} \left(\frac{\delta}{|\lambda-\lambda_i|}\right)^2\\
        &= \delta^2 \sum_{i\not\in A} \left(\lambda-\lambda_i\right)^{-2}\\
        &\leq \delta^2 \sum_{i\not\in A} \gamma^{-2}\\
        &= (\delta/\gamma)^{2} (n-m) 
    \end{align*}
    Thus $\|Pv-u\| \leq (\delta/\gamma) \sqrt{n-m}$.

    This tells us that in a partition on a graph that deviates from an AEP mildly, the eigenvectors with eigenvalues close to $\lambda$ describe $Pv$ well.
\end{proof}

This characterizes an eigenvector whose eigenvalue is close to that of $Pv$ as being a good approximation of it. As $\|\epsilon\|$ approaches zero, $\|Pv-u\|$ does as well for any truncated approximation $u$. In particular, when $\delta$ is small, then with a $\gamma$ small enough to only encompass a single eigenvalue, the corresponding  eigenvector approximates $Pv$ well. This shows a stability in the spectral properties of the almost equitable partition. We call these approximations of AEPs, quasi-equitable partitions ($\delta-$QEPs). In the specific, define a $\delta$-QEP or a QEP at level $\delta$, to be a partition $\pi$ such that $\|E\| \leq \delta$ where $E$ is the deviation matrix $PL^\pi - LP$. We show an example at $\delta=2.4$ in figure \ref{fig:delta_QEP_Cluster} A.
Whereas Kato and Ishii\cite{kato_cluster_2023} define their weakening in terms of a bound on the edge weight sum of vertices to partition cells, we define ours with the norm of a matrix capturing the deviation from the average, which allows it to maintain a relationship with the underlying spectral theory as we have demonstrated above.

\subsection{Weak Cluster Synchronization}
Because our description of cluster synchronization is spectral instead of combinatorial, we can immediately apply the same methodology from AEPs to the structural eigenvectors of $\delta-$QEPs to describe synchronized clusters (figure \ref{fig:delta_QEP_Cluster}). When moving from an AEP to a $\delta-$QEP, we can no longer construct the full eigenbasis of a corresponding quotient Laplacian. Non-structural eigenvectors have nonzero coefficients in the eigendecomposition of a natural frequency vector, $\vec{\omega}$, with constant values within clusters. Information is leaked into the rest of the eigenbasis, and it is precisely this leakage which results in differences in phase within clusters, as seen in figure \ref{fig:delta_QEP_Cluster} B. This is small if the $\delta-$QEP is reasonably close to an AEP. In this case the contribution of the non structural eigenvectors in the description of the natural frequencies results in a small but nonzero end state coefficient for them (figure \ref{fig:delta_QEP_Cluster} C.). Additionally, the structural eigenvectors are no longer completely constant within clusters. The non exact cluster synchronization arises from the difference in scaling moving from the decomposition of the natural frequency vector to the decomposition of the end state of the system.

\begin{figure}
    \centering
    \includegraphics[width=\linewidth]{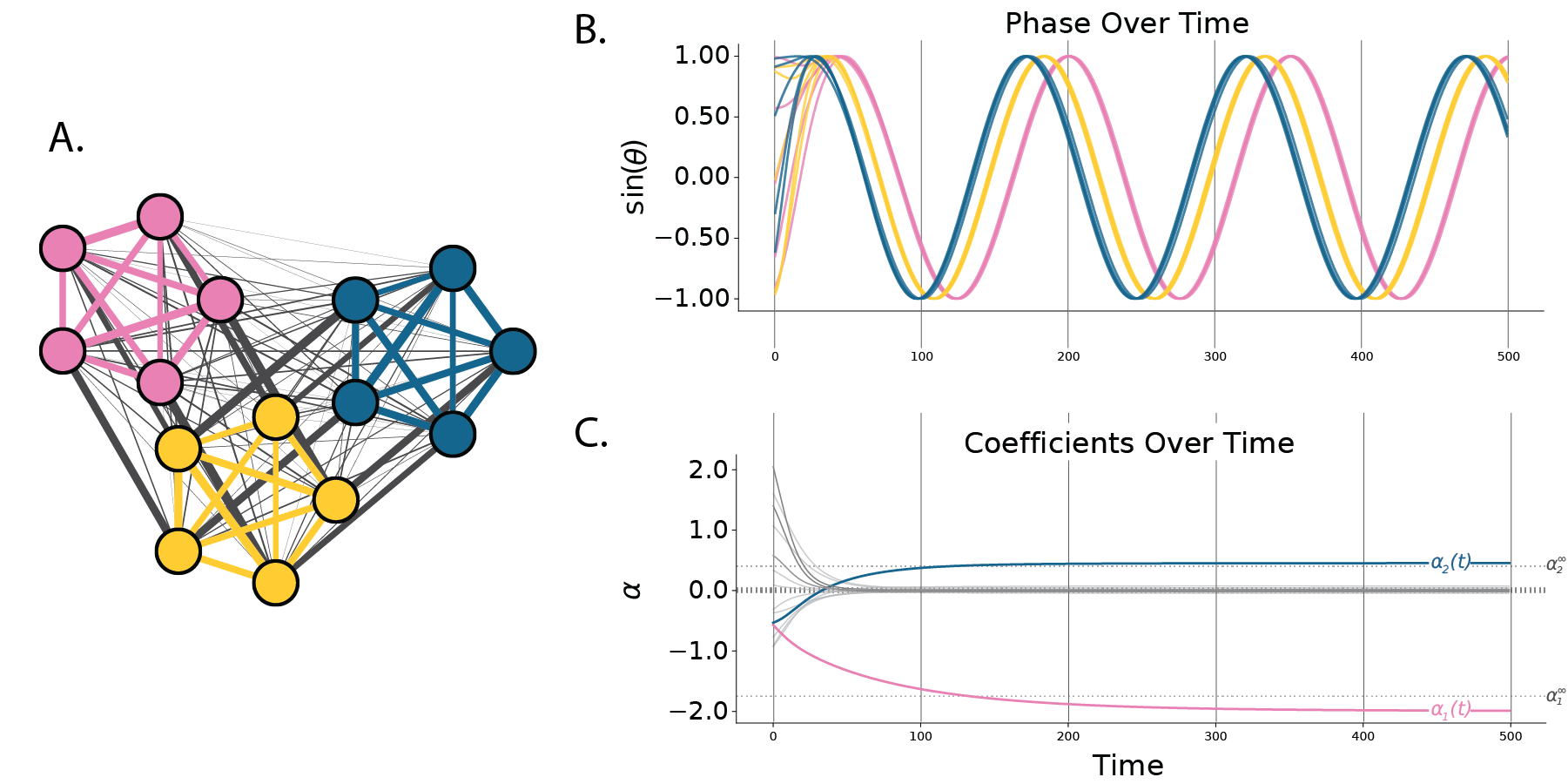}
    \caption{\textbf{A. The weighted network} consists of $15$ vertices and a color coded $\delta-$QEP at level $\delta=2.4$ which it admits. This graph was generated by adding uniform noise to the toy example from figure \ref{fig:2 AEP Phase decomp}.
    \textbf{B. The phase plot} indicates the phase locked cluster synchronization of the Kuramoto model placed on this graph with coupling $\sigma = 1.2$ and identical natural frequencies within each cell (respectfully the pink, yellow, blue cells have frequencies $3,4,5$). This frequency distribution reinforces cluster synchronization by mostly exciting structural modes of the $\delta-QEP$. \textbf{C.Time evolution of the coefficients in the projection onto the Laplacian eigenbasis.} The cluster synchronized state is described almost entirely by the quotient graph sub-basis, $\alpha_1(t)$ and $\alpha_2(t)$, with the leakage of information to the rest of the eigenbasis indicated by the spread of non structural coefficients around zero in C., and slight spread of phases within clusters in B.}
    \label{fig:delta_QEP_Cluster}
\end{figure}

In section \ref{sec:Evecbais} we discussed assortativity and spectrum placement with AEPs. The same argument using random walks can be extended to weighted graphs and $\delta$-QEPs.\cite{saade_spectral_nodate} The proximity of a $\delta$-QEP to a full AEP determines the bound of the variance in the intra cluster values of a structural eigenvector. With a reasonably small $\delta$, we can expect generically few zero crossings to occur within a cluster. Further, since proposition \ref{prop:Best_Approx} found that the best truncated approximation of an AEP structural eigenvector in the spectrum of a close $\delta-$QEP uses the eigenvectors with the nearest eigenvalues, we can expect the transient clustering discussed in section \ref{sec:transient} to behave well in the generalization to $\delta-$QEPs, as the modes needed to describe clusters at a particular scale have close eigenvalues to one another and will consequently decay at a similar rate.

\section{Multi frequency Clustering and Antisymmetric Phase lag}

The presented spectral description of clustered dynamics allows for inroads to other forms of cluster synchronization and extensions to related systems. We show an initial exploration of multi frequency clustering in this framework, then derive an extension of the coefficient equation to include antisymmetric phase lag within the network. 

\subsection{Clustered Dynamics at Separate Frequencies}

The description of cluster synchronization in terms of constant end-state coefficients explored to this point allows for clusters at offset phases, but at the same frequency. We now consider what a more general synchronization regime, where clusters advance at separate frequencies, may look like in this framework. Clustering along an AEP of the graph, this case is still fully described by a restriction to the structural sub-basis, but the coefficients may no longer approach constants. This more general clustered behavior would have structural coefficients as unstable and dynamical, while the nonstructural ones are kept close to zero.

An analytically tractable way to explore this through the coefficient equations is by considering the single unstable mode dynamics. 
Consider a coefficient $1\leq r_1\leq n-1$, and assume non zero modes $r\neq 0,r_1$ have settled into the asymptotic state predicted in the linearized case $\alpha_r^\infty$. We extend the derivation from Kallionatis\cite{kalloniatis_incoherence_2010} to weighted graphs. One can find parameter values of a regime change for the dynamical mode $\alpha_{r_1}$ in a second order approximation. Particularly we see the following, the derivation of which is in appendix \ref{sec:singlemodeweight}.

\begin{equation*}
    \dot{\alpha_{r_1}}(t) = \sigma\lambda_{r_1}\alpha_{r_1}(t) - \sigma x_{r_1} \alpha_{r_1}^2(t),
\end{equation*}
\begin{align*}
    x_{r_1} &= \sum_{s\neq 0,r_1} \frac{\omega^{(s)}}{2 \sigma \lambda_s}\sum_{a=1}^M W_{aa} (e_a^{(r_1)})^3e_a^{(s)}.
\end{align*}

Noting that the above is the logistic equation, the rest of the analysis follows the well known second order ODE framework. The discriminant of the homogeneous solution is
\begin{equation}\label{eqn:discrim}
    \Delta_{r_1}  = \sigma^2\lambda_{r_1}^2 - 4\sigma \omega^{(r_1)}x_{r_1}.
\end{equation}
The variation of the parameter $\sigma$ allows for $\Delta_{r_1}$ to switch signs accompanying a qualitative shift in the behavior of the system, when $\Delta_{r_1}>0$ the solution approaches a fixed point, whereas for $\Delta_{r_1}<0$, the dynamics approach a limit cycle instead. The stable solution for $\Delta_{r_1}>0$ is
\begin{align*}
    \alpha_{r_1}(t) = 
    \frac{(\sigma \lambda_{r_1} + \sqrt{\Delta_{r_1}}) - (\sigma\lambda_{r_1} - \sqrt{\Delta_{r_1}})P_{r_1}e^{\sqrt{\Delta_{r_1}}t}}{2\sigma x_{r_1} (1-P_{r_1}e^{\sqrt{\Delta_{r_1}}t})}
\end{align*}
whereas for $\Delta_{r_1}<0$ the dynamics can be approximated by the following equation within reasonable bounds of the origin.
\begin{align}\label{single_mode_approx}
    \alpha_{r_1}(t) = \frac{1}{2\sigma x_{r_1}}\left( \sigma\lambda_{r_1} + \tan\left\{ \arctan\left[ 2\sigma x_{r_1} \alpha_{r_1} - \sigma \lambda_{r_1}\right] + \sqrt{-\Delta_{r_1}}\frac{t}{2} \right\} \right)
\end{align}

This second order approximation and an accompanying equilibrium analysis suggest limit cycle behavior with the mode $r_1$ when $\Delta_{r_1}<0$. This approximation can yield high values due to the divergence of tangent near $\pi/2$, pulling the system outside the reasonable scope of its approximation. 

Notice that in order to have $\Delta_{r_1}<0$, the value $\sigma\lambda_{r_1}$ must be sufficiently small. This suggests that eigenmodes with lower eigenvalues will more easily support the proposed single periodic dynamical mode behavior. We also note that $\omega^{(r_1)}$ shows up with a negative sign in equation \ref{eqn:discrim}, meaning that for large $\omega^{(r_1)}$ the discriminant is negative. Thus a mode is more able to remain dynamical in this regime if its alignment with the natural frequency vector is high. We see the discriminants for the modes of an example system plotted in figure \ref{fig:single_mode_dynamic}.

\begin{prop}
    Consider a single dynamical mode, $r_1$ as in the above derivation, corresponding to an AEP structural eigenvector without a repeated eigenvalue. Let $\vec{\omega} = a v_{r_1}$ for $a\in \R$. Under these assumptions, the system experiences clustered behavior with at least two clusters at separate instantaneous frequencies.
\end{prop}
\begin{proof}
    Let $\pi= \{V_1,\dotsc,V_k\}$ be the AEP inducing the structural $v_{r_1}$. Since $\vec{\omega} = a v_{r_1}$, it follows that $\alpha_r^{\infty}=0$ for all $r\neq r_1,0$. Consequently $\vec{\theta}= \alpha_{r_1}(t) v_{r_1}$, which implies $\dot{\vec{\theta}} = \dot{\alpha}_{r_1}(t) v_{r_1}$. 

    It follows that $\dot{\theta}_i(t) =\dot{ \alpha}_{r_1}(t) (v_{r_1})_i$. By the assumption on $v_{r_1}$, it follows that $(v_{r_1})_i= (v_{r_1})_j$ for all $i,j\in V_p$, where $1\leq p\leq k$. Thus all vertices within a cluster experience the same angular velocity. Since $r_1\neq 0$, there exist two partition cells such that their corresponding values in the vector $v_{r_1}$ differ. Thus for all $i,j\in V_p$ with $1\leq p\leq k$, we have $\dot{\theta}_i(t) = \dot{\theta}_j(t)$, and for some pair of $1\leq p, q\leq k$ we have $\dot{\theta}_i(t) \neq \dot{\theta}_j(t)$ for any $i\in V_p$ and $j\in V_q$.

    Dividing the angular velocities by $2\pi$, we obtain the desired result for the angular frequency.
\end{proof}

% We note the difference of the above result 

In our testing, we could observe that the approximation described in eq. \ref{single_mode_approx} holds well with AEPs. See the example system in figure \ref{fig:single_mode_dynamic}, with $\Delta_{1}<0$. here we see the system approach a cluster synchronized state quickly and stay cluster synchronized thereafter. The dynamical mode is well approximated by the tangent curve in equation \ref{single_mode_approx} for 2000 time steps, before the predicted divergent behavior activates the structural $\lambda_2$ mode. The system then converges back to the single mode regime of our approximation before repeating. The stability of the periodic behavior in this test is motivated by only the $\lambda_1$ mode achieving a negative discriminant. We see the momentary perturbation of $\lambda_2$ whose discriminant was the 2nd smallest but still positive, suggesting a return to a stable steady state.

This shows that the lower spatial frequencies can support dynamical modes better. In the context of our analysis of structural eigenvectors, this suggests that the eigenvectors of assortative AEPs or $\delta-$QEPs allow for unstable mode behavior more easily. Thus clusters with strong intra cluster connectivity and weak inter cluster connectivity should be able to maintain separate cluster frequencies if the disparity in connectivity is high enough and the natural frequencies align well enough with one of the structural eigenvectors. This intuition also suggests that an analogous property should hold for multiple unstable modes, if they all share reasonably high alignment with the natural frequency vector.

While the discriminant was derived  explicitly for the case of a single dynamical mode, our experiments suggest that it can serve to gauge the capacity of any single mode to be dynamical even in the presence of other dynamical modes. Although we do not have an analytical approximation for it, we find cases in which multiple modes are unstable and share the repeated stair-step, tangent-like form of the single unstable mode. With negative discriminants on two of the lower eigenvalue modes, we observe no activation of external modes and complete restriction to the two unstable and coupled modes. This suggests a stronger stability in the restriction to the two-mode dynamics than with the single mode described above.

\begin{figure}
    \centering
    \includegraphics[width=\linewidth]{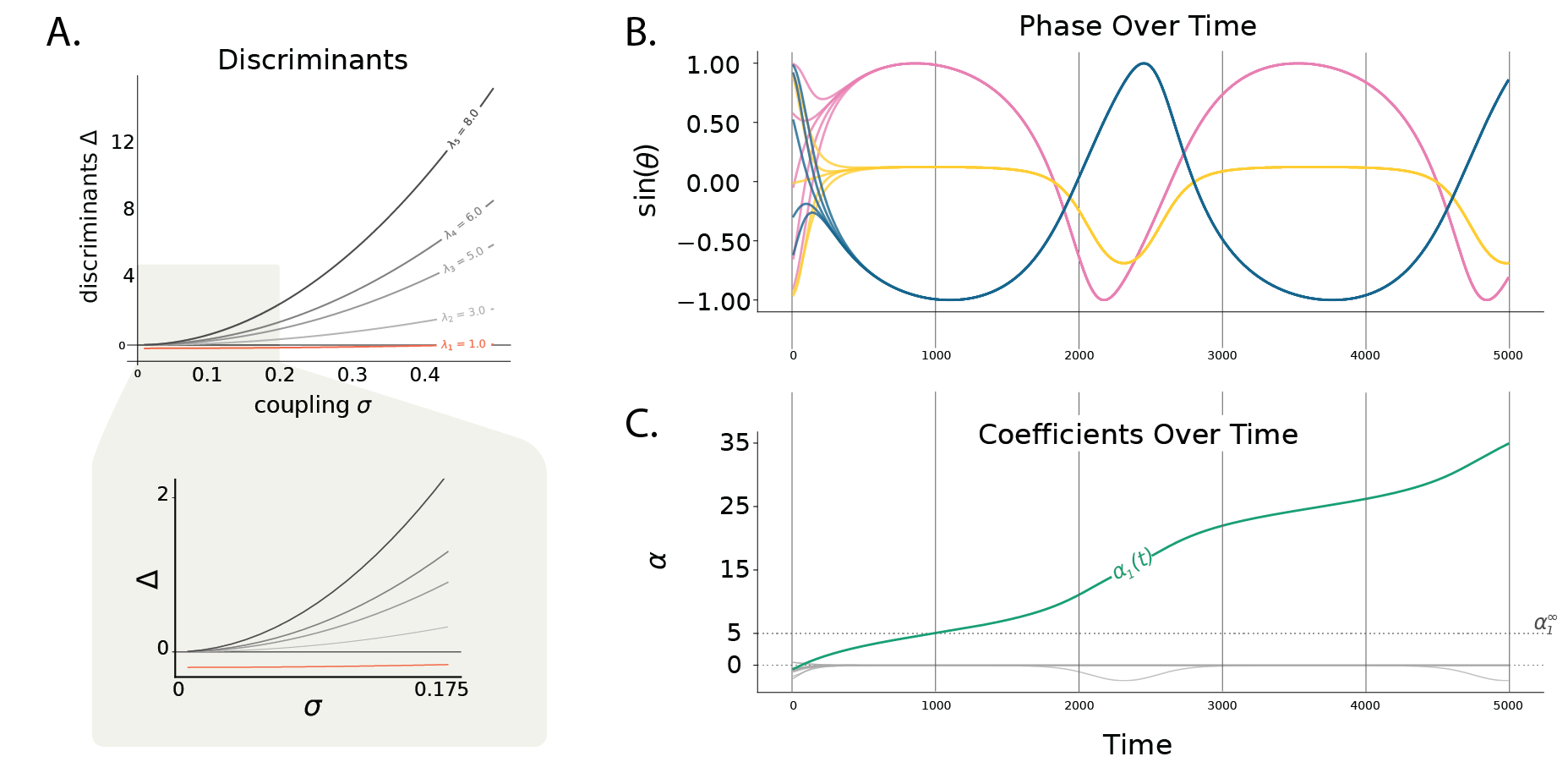}
    \caption{We dive deeper into clustered dynamics of the sample graph in Fig.\ref{fig:2 AEP Phase decomp}. \textbf{A. The discriminants} $\Delta_{r}$ for all eigenmodes $r$ as functions of the coupling $\sigma$, with fixed natural frequencies, constant within clusters ($\vec{\omega}$ is a multiple of the eigenvector $v_1$). \textbf{B. The phase plots} show the system experiencing multi-frequency clustered behavior with a choice of $\sigma=0.2$, inside the negative range of $\Delta_1$. Finally,
    \textbf{C. depicts the coefficients of the phase decomposition}, indicating a single dynamical mode, $\alpha_1(t)$. Note the periodic shifts in $\alpha_1(t)$ along with the activation and decay of a second mode ($\alpha_2(t)$). The periodic jumps in $\alpha_2(t)$ are reflected in the periodic changes in the clustered dynamics in \textbf{B}, but we note that while this is a periodically recurrent change, the clustered oscillators do experience a different number of oscillations from one another. In this case, the yellow cluster never completes a full oscillation, and the other two clusters are in counter rotation to one another.}
    \label{fig:single_mode_dynamic}
\end{figure}
This reasoning does not extend easily to $\delta-$QEPs. Certainly the machinery works for any eigenvectors of the graph Laplacian, but the information leak of the structural eigenvectors for a $\delta-$QEP complicates the description of clustered behavior. Given a system with constant natural frequencies within clusters, the structural eigenvectors of an associated $\delta-$QEP generically won't be able to fully describe the clustered dynamics. In the frequency locked synchronization regime, this resulted in small activation from other modes. When considering a single unstable mode, that information from non structural modes can allow them to also become unstable. This behavior distances the dynamics from the above analysis, and makes the results less interpretable. However, the clustered behavior is still present, and a more catered choice of natural frequencies should allow for the same behavior in isolated coefficients.

\subsection{Spectral Cluster Synchronization With Phase Lag}\label{sec:phase lag}

We now present an extension to the coefficient form of Kuramoto with the addition of antisymmetric phase lag between oscillators. The following Kuramoto Sakaguchi equations describe a well known generalization of the Kuramoto model to include phase differences $\beta_{ij}$ along edges
\begin{equation}\label{sakaguchi}
    \dot{\theta_i} = \omega_i - \sigma \sum_{j\in V} A_{ij} \sin(\theta_i - \theta_j+ \beta_{ij}).
\end{equation}
We consider the case in which $\beta_{ij}=-\beta_{ji}$. Recall that in constructing the coefficient equation we used the signed incidence matrix $B$ with some arbitrary orientation on the edges. For edge $a$ with orientation $(i,j)$ in $B$ we denote the phase difference along $\beta_{ij}$ as $\beta_a$ in the following extension of the coefficient equation:

\begin{equation}\label{coeff_phase_lag}
    \dot{\alpha}_r = \omega^{(r)} - \sigma \sum_{a\in E(G)} W_{aa}   e_a^{(r)} \sin\left( \sum_{s>0} e_a^{(s)}\alpha_s(t) + \beta_{a}\right).
\end{equation}

The linearized equilibrium analysis reveals that close to equilibrium, and with small $\beta_a$, we get the following limit for the coefficients:

\[
\alpha_r^{\infty} =\frac{\omega^{(r)}-\sigma \sum_{a\in E(G)}W_{aa}e_a^{(r)}\beta_{a}}{\sigma \lambda_r}.
\]

The interested reader can find a full derivation of these equations in appendix \ref{apdx:Coeff_Lapl}. Let's consider a couple examples to put this equation in context.

\textbf{Example 1:} Suppose the Kuramoto Sakaguchi model was placed on a network with an AEP, where intra cluster edges exhibit phase lag $\beta_a>0$, and inter cluster edges have negligible phase lag with $\beta_a\approx 0$. As before, natural frequencies are constant within clusters. In this case it follows that the equilibrium for all coefficients is:
\begin{align*}
    \alpha_r^{\infty} &=\frac{\omega^{(r)}-\sigma\sum^{\text{intra}}_{a}W_{aa}e_a^{(r)}\beta_{a}}{\lambda_r \sigma}.\\
\end{align*}
Then structural coefficients for our AEP approach the limit $
\frac{\omega^{(r)}}{\lambda_r \sigma},
$
whereas for the non structural coefficients approach
$
-\frac{\sum^{\text{intra}}_{a}W_{aa}e_a^{(r)}\beta_{a}}{\lambda_r \sigma}.
$
So in this example we see the structural coefficients approach the same values we would expect in a system without phase lag, exhibiting an isolation of the information for clustering behavior. Then the non structural coefficients, which dictate our distance from full cluster synchronization, are excited and reach their own equilibrium through an expression of the intra cluster phase lag and edge weights.  

Thus we see exactly what our intuition would suggest: a separation into the modes describing cluster synchronization and the modes adding noise, which spring up entirely due to the phase offset parameters and weighted edges. Note that this noise is also entirely independent of the coupling strength.\\

\textbf{Example 2: } Another telling example to consider is a network with uniform edge weight $w$. In this case, the equilibrium state becomes
\begin{equation*}
    \alpha_r^{\infty} =\frac{\omega^{(r)}-w\beta^{(r)}}{\lambda_r \sigma}
\end{equation*}
where we define $\beta^{(r)}=e^{(r)}\cdot \vec{\beta}$. This is the alignment of the $r$th down edge Laplacian eigenvector with the phase offset vector. Assuming this network has an AEP and that the natural frequencies are constant within clusters, recall that the down edge Laplacian eigenvectors can each be interpreted as providing an edge value assignment. Now suppose that $\vec{\beta}$ is a linear combination of the vectors $e^{(r)}$ corresponding to the structural eigenvectors. This requires equality in the antisymmetric phase-lag between pairs of clusters, with respect to the direction, and zero phase-lag within clusters. In this case, for structural coefficients, the equilibrium state looks as above, but by orthogonality, the nonstructural coefficients have $\beta^{(r)}=0$. Thus, for nonstructural eigenmodes, $r$, we get $\alpha^{\infty}_r = 0$, which means we can recover cluster synchronization in case of phase lag. Intuitively, our example requires the phase lag to respect the vertex partition. 
Moreover, when we allow edge weights to vary but keep a constant phase lag across all edges, the equilibrium state for the coefficients becomes
\begin{equation*}
    \alpha_r^{\infty} =\frac{\omega^{(r)}-w^{(r)}\beta}{\lambda_r \sigma}.
\end{equation*}
Here with $\vec{w}$ as the vector of edge weights we define $w^{(r)}=e^{(r)}\cdot \vec{w}$, the alignment of the $r$th down edge Laplacian eigenvector with the edge weight vector. This reveals an interesting duality between constant phase lag and constant weight, where they have comparable impacts on synchronization.

\section{Discussion}\label{discuss}
We draw from spectral graph theory and recent developments in network synchronization research to propose a generalized spectral framework that extends traditional synchronization analysis. We show how this new framework can be used to rewrite complex dynamical systems to better study the formation and evolution of clustered behaviors over a network. However, there are a couple considerations to keep in mind, indicating the use cases for which our spectral approach is best suited.

The first consideration is the proximity to synchronization, as transient behavior can disrupt the change of basis. The phases of our oscillators are in $\R$ mod $2\pi$, but the basis of eigenvectors represents vectors valued in $\R$. Thus, it is possible for transient behavior on the way to synchronization to result in phase offsets in multiples of $2\pi$ between synchronized oscillators (a non-trivial winding cell). These are equivalent phases mod $2\pi$, but not in the eigenvector basis, which results in the activation of spurious modes being necessary to describe the differences. This non-equivalence in the eigenvector basis means that we lose interpretability when the initial dynamics allow for some oscillators to get a full rotation ahead or behind. If this occurs at some point before synchronization, one could re-zero the phase vector closer to the synchronization to regain the interpretability of the decomposition thereafter. However, this property suggests that instead of chaotic system wide behavior, the eigenvector basis is most suited to studying synchronization patterns, or fully clustered dynamics, as we have in this paper.

Our work has a clear relationship to spectral clustering techniques and the study of cluster detectability in networks. It is easy to show, for example, that a stochastic block model (SBM)  does approach an AEP as the number of vertices in the graph approaches infinity.
The convergence to an AEP follows from showing the null expectation of the values in $E$, and the law of large numbers (See the appendix \ref{app:SBMlimit} for more of the proof). 

\begin{prop}
    In the positive limit of $n$, there exists a natural vertex partition for sampled graphs, given by the SBM block structure, for which the deviation from the AEP condition (or more precisely, EP) vanishes in the large-size limit.
\end{prop}

This relates our notion of a $\delta$-QEP to the previously developed model of stochastic AEPs, which describes simple graphs which are almost equitable in expectation.\cite{schaub_hierarchical_2023}

Recently, structural eigenvectors of the Bethe Hessian have been used for spectral clustering algorithms in sparse networks.\cite{duan_spectral_2020} We note that this matrix is a deformed Laplacian, sharing the same eigenvectors, but with rescaled eigenvalues. Thus one can study the same structural eigenvectors\cite{krzakala_spectral_2013} in relation to $\delta-$QEP community structure using 1 dimensional clustering\cite{patania_exact_2023} along individual eigenvectors or $k-$means on a block of eigenvectors to separate community structure.

We note that graphs with high symmetry can allow for duplicate AEPs, which manifest in the eigenbasis with repeated eigenvalues. These repeated eigenvalues allow for variability in the choices of eigenvector representatives. In this case eigenvector representatives of one associated AEP may not have uniform values within clusters. This slightly changes how the results are interpreted, but no matter the choice of eigenvector representative, the sub-basis of structural eigenvectors is still formed, allowing for the dimensionally reduced description of clustered dynamics. 
This kind of situation may lead to difficulty with detectability using 1-dimensional clustering methods because the choices of eigenvector representatives affect which group structure is displayed. However, in practice this effect does not have a substantial impact as most networks don't exhibit such high regularity, and this is even rarer for weighted networks.

While the focus of this work is theoretical, the spectral framework developed here has direct implications for the analysis of real-world systems where synchronization plays a functional role. In neuroscience, for example, both structural and functional brain networks exhibit modular and near-symmetric organization, which can support cluster synchronization across regions. Studies have shown that Laplacian eigenmodes align with empirically observed patterns of resting-state brain activity,\cite{atasoy_human_2016, deco_emerging_2011, cabral2011role} suggesting that AEP- or $\delta$-QEP-based reductions could provide mechanistic insight into these patterns.

An exciting direction for future work involves extending this spectral approach from vertex dynamics to edge dynamics. Recent work by Faskowitz et al. has highlighted the role of edge-centric representations in brain connectivity,\cite{faskowitz2020edge} and Maria Pope and colleagues have explored dynamic synchronization in networks where edge properties (e.g., weights or phase lags) change over time in response to activity.\cite{pope2023co} Adapting our framework to such edge-based models may provide new tools for capturing coordination at finer spatiotemporal resolutions and for modeling dynamic reconfiguration in systems where edge dynamics play an active role in shaping global synchrony.

Similarly, in power grid networks, synchronization of coupled oscillatory units (e.g., generators and loads) is essential for operational stability. The grid’s underlying modular structure often gives rise to approximate symmetries that could be captured by $\delta-$QEPs.\cite{Chow_Power_Grid_Coherency} Previous work has used spectral and symmetry-based approaches to study synchronization and stability in such systems,\cite{dorfler2012synchronization, motter2013spontaneous, nishikawa2015comparative} and our framework may provide an additional low-dimensional tool for modeling and control in this context. A natural next step is to integrate this framework into empirical studies of these systems, further bridging the gap between network structure and dynamic behavior.

In general, many real-world systems, such as biological or social networks, include inhibitory or antagonistic interactions, naturally leading to signed networks. The current framework assumes all edge weights are strictly positive; this assumption is primarily made to preserve the positive semi-definiteness of the graph Laplacian and to ensure convergence properties of the Kuramoto dynamics. Extending the spectral framework to accommodate such systems would require analyzing the Laplacian of signed graphs. This introduces new challenges, such as the possibility of negative eigenvalues and less predictable anti-synchronization behavior. We see this as an important avenue for future research, particularly in modeling competitive interactions and studying how inhibitory coupling affects the formation and stability of clustered dynamics.

\section{Conclusion}

In this work, we draw from multiple fields to build and present a novel approach to the study of clustered dynamics on graphs.
While much of the previous work studying cluster synchronization of the Kuramoto model has focused on structural properties in the graph itself, we show how these classical approaches relate to the field of signal processing on graphs, spectral graph theory, and community detection. 

Although our dynamical model is defined using the adjacency matrix, we emphasize that the Laplacian eigenvectors provide a particularly suitable basis for analyzing synchronization dynamics. This is because the Laplacian as an operator captures diffusion over the network, and its eigenbasis naturally encodes both global coherence (through the zero eigenmode) and modular organization (through low-lying nontrivial eigenvectors). The Kuramoto model is a diffusive system and its linearization can be expressed directly in terms of the Laplacian. More generally, the Laplacian spectrum reveals how perturbations propagate and decay across the network, offering a principled way to identify and interpret emergent synchronizing clusters. Thus, using the Laplacian as the analytic backbone of our framework ensures that we capture both the structural and dynamical factors shaping clustered synchronization.

We expand on how the eigenbasis of the graph Laplacian encodes information fundamental to the synchronization of network models, and can be used to reformulate the Kuramoto model's equations. This theoretical advancement reveals how the structure encoded in these eigenvectors provides direct insight into cluster synchronization and clustered dynamics, particularly when examining network vertex partitions into synchronizable clusters.

Our framework recontextualizes AEP cluster synchronization as the dynamics of structural eigenvectors of the graph laplacian. We introduce $\delta$ level quasi-equitable partitions ($\delta-$QEPs) as an approximate generalization of AEPs, allowing for deviations from the AEP spectrum. We prove that the spectral properties we establish for AEPs extend to $\delta-$QEPs when variations are sufficiently small, and we characterize the stability of these results as $\delta-$QEPs diverge from AEPs.

Furthermore, we demonstrate the versatility of our spectral approach by extending it to variants of the Kuramoto model, including networks with weighted edges and antisymmetric phase lag. By unifying multiple theoretical frameworks for studying synchronization, our work not only bridges existing approaches but also opens new avenues for investigating complex network dynamics. This unified spectral perspective provides a promising foundation for future research in networked dynamical systems and their synchronization properties.

\begin{acknowledgments}
We thank the Reviewers and the Editor for their many helpful suggestions and comments.

This material is based upon work supported by the National Science Foundation Graduate 
Research Fellowship Program under Grant No. 2235204. Any opinions, findings, 
and conclusions or recommendations expressed in this material are those of the authors 
and do not necessarily reflect the views of the National Science Foundation.
\end{acknowledgments}

\section*{Author Declarations}
The authors have no conflicts to disclose

\section*{Data Availability Statement}
The data and code that support the findings of this paper, and recreate the figures, are openly
available on github with the following link \url{https://github.com/TobiasTimofeyev/Spectral_Kuramoto}\cite{timofeyev2025spectral}.

\appendix

\section{Coefficient Form of Kuramoto}\label{apdx:Coeff_Lapl}
    We show the derivation of equations \ref{coeff_phase_lag} (and \ref{coeffeqn} by extension) onto weighted networks with antisymmetric phase lag, extending the equations as derived by Kalloniatis.\cite{kalloniatis_incoherence_2010} Our antisymmetric phase lag assumption means that the phase offset adheres to $\beta_{ij}=-\beta_{ji}$. As we saw in section \ref{sec:graphMat}, the graph Laplacian of a weighted graph can be written as $L=BWB^T$, where $W$ is a diagonal matrix of positive edge weights, and $B$ is the signed incidence matrix. As we saw with the unweighted Laplacian, there is an edge specific counterpart with which it shares a spectrum. Consider
    \[
    Lv=\lambda v \implies BWB^T v= \lambda v\implies B^TBW\left(B^Tv\right)=\lambda \left(B^T v\right)
    \]
    and
    \[
    B^TBW u=\lambda u \implies W B^TB u=\lambda u \implies  B W B^T (B u)=\lambda (B u)
    \]
    This establishes a one to one correspondence between eigenvector value pairs of $L=BWB^T$ and $B^TBW$. Denote the $r$th eigenvector of $B^TBW$ as $e^{(r)}$.
    Let $A$ be the weighted adjacency matrix. Let $\vec{\theta}=\sum_r \dot{\alpha_r}(t)v^{(r)}$ be the eigendecomposition of $\theta_i$ into the eigenvectors of $L=BWB^T$. Given the Kuramoto Sakaguchi equations with antisymmetric phase offsets, we can perform the following change of basis. 
    \begin{align*}
        \dot{\theta_i}(t)&=\omega_i-\frac{K}{N}\sum_{j=1}^N A_{ij} \sin(\theta_i-\theta_j+ \beta_{ij})\\
        \implies \sum_r \dot{\alpha_r}(t)v_i^{(r)}&=\omega_i-\frac{K}{N}\sum_{j=1}^N A_{ij} \sin\left(\sum_{s>0} \alpha_s(t) (v_i^{(s)}-v_j^{(s)}) +\beta_{ij}\right)\\
        \implies \sum_{r} \dot{\alpha_r}(t)v^{(r)}&=\omega-\frac{K}{N}\sum_{a=1}^M B_{a} W_{aa} \sin\left(\sum_{s>0} \alpha_s(t) B_a^T v^{(s)} + \beta_a\right)
    \end{align*}
    The last step turns the equation into a vector equation, and the sum from one over neighbors to one over incident edges. We denote the $a$th column of $B$ with $B_a$. This column corresponds to an edge, and contains a single $1$ and $-1$ to denote the choice of edge orientation. Note that two of these vectors are present in the vector formulation. Since sine is an odd function, the negative value in the $B_a$ outside flips the orientation of the difference $v_i^{(s)}-v_j^{(s)}=B_a^Tv^{(s)}$ to the one appropriate for that vertex. The antisymmetry of the phase offset comes into effect here. For an edge $\{i,j\}$ with the orientation for the signed incidence as written, and identified as $a$, we let $\beta_a=\beta_{ij}$. Then $\beta_{ji}=-\beta_a$, which allows this offset to agree with the sign flip inside the sine function. We may now isolate the coefficient using the orthogonality of the basis.
    \begin{align*}
        \dot{\alpha_r}(t)&=\omega\cdot v^{(r)}-\frac{K}{N}\sum_{a=1}^M v^{(r)}\cdot B_{a} W_{aa} \sin\left(\sum_{s>0} \alpha_s(t) B_a^T v^{(s)} + \beta_a\right)\\
        &\text{by orthogonality}\\
        \implies \dot{\alpha_r}(t)&=\omega^{(r)}-\frac{K}{N}\sum_{a=1}^M e^{(r)}_a W_{aa} \sin\left(\sum_{s>0} \alpha_s(t) e_a^{(s)} + \beta_a\right)\\
        \alpha_0 &= \sqrt{N} \overline{\omega} t
    \end{align*}
    
    Here we define $\overline{\omega}$ to be the average of $\omega$, and $\omega^{(r)}=\omega\cdot v^{(r)}$. This concludes the derivation of the coefficient equation with edge weights and antisymmetric phase offsets. Note that the choice of $\beta_a$ for each edge $a$ corresponds to the value at the orientation of the signed incidence matrix, and the value of the opposing direction simply has an opposite sign. In this way the chosen orientations are arbitrary and have no impact on the dynamics.
    
    We now proceed to the corresponding steady state analysis. To begin with, we show a property of the vectors $e^{(r)}$. 
    
    \begin{prop}
        Let $(e^{(r)}, \lambda_r)$ and $(e^{(s)},\lambda_s)$ be eigenvector-value pairs of the down edge Laplacian $L^\text{DN}=B^TBW$. It follows 
        \[   (e^{(r)})^T W e^{(s)}=
        \begin{cases} 
          \lambda_s & \hspace{10pt} \text{if $r=s$}\\
          0     &\hspace{10pt}   \text{otherwise.}
       \end{cases}
        \]
    \end{prop}
    \begin{proof}
        Let $L=BWB^T$ be the corresponding graph Laplacian. It follows that $e^{(r)}=B^Tv^{(r)}$ and $e^{(s)}=B^Tv^{(s)}$ where $v^{(r)}$ and $v^{(s)}$ are eigenvectors of $L$ with eigenvalues $\lambda_r$ and $\lambda_s$ respectively. It follows
        \[
        (e^{(r)})^T W e^{(s)}=(B^Tv^{(r)})^T W B^Tv^{(s)}=(v^{(r)})^T BWB^Tv^{(s)}=(v^{(r)})^T v^{(s)}\lambda_s.
        \]
    The proof concludes by the orthogonality of the eigenvectors of the graph laplacian $L$.
    \end{proof}
    
    The above is not quite orthogonality of the eigenvectors of $L^\text{DN}$, since in this weighted formulation it loses symmetry, but it serves a similar purpose. Consider the system close to synchronization for $r\not=0$ and relatively small phase offsets. This supports a linearized approximation of the sine term.

    \begin{align*}
        \dot{\alpha_r}(t)&=\omega^{(r)}-\frac{K}{N}\sum_{a=1}^M e^{(r)}_a W_{aa} \sin\left(\beta_a + \sum_{s>0} \alpha_s(t) e_a^{(s)}\right)\\
        &\approx \omega^{(r)}-\frac{K}{N}\sum_{a=1}^M e^{(r)}_a W_{aa} \left(\beta_a + \sum_{s>0} \alpha_s(t) e_a^{(s)}\right)\\
        &= \omega^{(r)}-\frac{K}{N}\left(\sum_{a=1}^M  e^{(r)}_a W_{aa} \beta_a - \sum_{s>0}\alpha_s(t)\sum_{a=1}^M  e^{(r)}_a W_{aa}   e_a^{(s)}\right)\\
        &=\omega^{(r)}- \frac{K}{N} \sum_{a=1}^M  e^{(r)}_a W_{aa} \beta_a -\frac{K}{N}\alpha_r(t)  \lambda_{r}\\
        &=\omega^{(r)}- \sigma \sum_{a=1}^M  e^{(r)}_a W_{aa} \beta_a -\sigma\alpha_r(t)  \lambda_{r}\\
    \end{align*}
    
    The solution of this linear ODE is standard and has the following form.
    
    \[
    \alpha_r(t) = \frac{\omega^{(r)} - \sigma \sum_{a=1}^M  e^{(r)}_a W_{aa} \beta_a}{\sigma \lambda_r} \left(1-e^{-\sigma \lambda_r}\right) + \alpha_r(0) e^{-\sigma \lambda_r}
    \]
    
    It follows that the asymptotic state of this linearized solution is
    
    \[
    \alpha_r^{\infty} =\frac{\omega^{(r)}-\sigma \sum_{a}W_{aa}e_a^{(r)}\beta_{a}}{\sigma \lambda_r}.
    \]

    Something to note is that the sine term in the original Kuramoto model is called the phase interaction function (PIF),\cite{breakspear_generative_2010} and all of our analysis would work with other PIFs as long as they are odd functions, with only mild adjustments if the first order taylor approximation differs.

\section{Transient Clustering}\label{apdx:transient}

We begin with the proof of proposition \ref{prop:transient}.

\begin{proof}
    Let $G$ be a connected graph with an AEP, $\pi$. Suppose $\pi$ is sufficiently assortative in that the $p\geq 1$ smallest eigenvalues of the graph Laplacian, $L$, are associated to corresponding structural eigenvectors. Assume there are no repeated eigenvalues. Consider a Kuramoto model on $G$ with uniform natural frequencies, that is $\vec{\omega} = c\mathbb{1}_n$ for some constant $c$.

    Recall the explicit solution for the linearized equations, \ref{eqn:linearcoeff}. Since $G$ is connected it follows that the eigenvalue $\lambda_0=0$ of $L$ has multiplicity one, and $\vec{v}_0=\frac{1}{\sqrt{N}}\mathbb{1}_n$. By orthogonality it follows that $\omega^{(r)}=0$ for all $0< r\leq n$. Thus equation \ref{eqn:linearcoeff} reduces to
    \[
        \alpha_r(t) = \alpha_r(0) e^{-\sigma \lambda_r t}
    \]
    for all $0< r\leq n$. Since $\sigma$ is a constant, the positive eigenvalue distinguishes the relative decay of differing eigenmodes fromt their starting values. By assumption, for modes $0<r\leq p$, the eigenvectors are structural, and without repeated eigenvalues, are in the column space of $P^\pi$. Thus when these modes dominate the dynamics of $\vec{\theta}$, cluster synchronization is achieved. Notice that for $0<i\leq p$ and $p<j\leq n$,
    \[
    \frac{|\alpha_i(t)|}{|\alpha_j(t)|} = \frac{|\alpha_i(0)|}{|\alpha_j(0)|} e^{-\sigma (\lambda_i-\lambda_j)t}.
    \]
    By assumption $\lambda_i<\lambda_j$, which implies that this ratio is growing with time for all such pairs of modes where $|\alpha_i(0)|,|\alpha_j(0)|>0$. Thus, naturally the eigenmodes with indices $0< i\leq p$ and non-zero initial condition will dominate with time. 
    
    Note that the choice of initial condition has a significant impact on the visibility of the clustered behavior, as the structural eigenmodes may only dominate once all coefficient magnitudes have become small.
\end{proof}

We note that with additional restrictions for repeated eigenvalues and eigenvectors, the above proof can straightforwardly be extended to the case of repeated eigenvalues. The clusters of the above emerge from a difference in the rates of decay. Since repeated eigenvalues create ambiguity in the choices of eigenvectors to span their induced subspace, a simple requirement that the choices of structural eigenvectors of the $p$ strictly smallest eigenvalues exist in the column space of $P^\pi$ is required.

Next we prove a lemma towards the proof of \ref{prop:transienthier}.

\begin{lemma}\label{lem:refinement}
    Consider two AEPs, $\pi_1$ and $\pi_2$, such that the former is a refinement of the latter. Then the column space of $P^{\pi_2}$ is contained in the column space of $P^{\pi_1}$.
\end{lemma}

\begin{proof}
    Let $w_i$ be the $i$th column of $P^{\pi_2}$. This has entries of $1$ for all vertex indices within cell $i$ of $\pi_2$. Since $\pi_1$ is a partition of the vertices, all the vertices of cell $i$ in $\pi_2$ are contained in some cell of $\pi_1$. Since $\pi_1$ is a refinement, $i$ is either a block in $\pi_2$ in its entirety, or subdivided into smaller blocks in $\pi_1$. In either case there exists some set of columns $u_1,\dotsc, u_p$ such that $w_i = u_1 + \dots + u_p$. Thus all the columns of $P^{\pi_2}$ are in the columns space of $P^{\pi_1}$ and we are done.
\end{proof}

Now, the proof of proposition \ref{prop:transienthier} follows from the previous lemma and proposition \ref{prop:transient}.

\begin{proof}
    We assume that some low structural modes for each AEP in the hierarchy have nonzero initial condition. Under proposition \ref{prop:transient}, in time the eigenvalues with indices below $p_k$ will dominate the rest to whatever degree desired. By lemma \ref{lem:refinement}, the combination of the $p_i$ lowest eigenvectors is in the column space of $P^{\pi_i}$, corresponding to the finest partition. Thus behavior dominated by the $p_k$ lowest eigenmodes corresponds to clustered behavior on $\pi_k$. However, proposition \ref{prop:transient} may be applied again to the $p_{k-1}$ smallest modes, and the $p_{k-2}$ smallest modes, and so on. Thus, up to the choice of initial conditions, the system can experience cluster synchronization corresponding to each level of the AEP hierarchy, on the way to full system synchronization.
\end{proof}

\section{Single mode dynamics for weighted graphs}\label{sec:singlemodeweight}

Consider the coefficient equation as written in equation \ref{coeffeqn}. We assume that the system has reached the equilibrium \ref{eqn:limitcoeff} in all modes except $r=r_1$ which remains time dependent. Given these assumptions, we approximate the dynamics of the coefficient equations on a weighted graph, to second order. Note that we drop cubic terms in the $r_1$ coefficient and quadratic in the others.
\begin{align*}
    \dot{\alpha}_{r_1}(t) &= \omega^{(r_1)} - \sigma \sum_{a=1}^M e_a^{(r_1)}W_{aa} \sin\left(\alpha_{r_1}(t)e_a^{(r_1)} + \sum_{s\neq 0,r_1}\alpha^\infty_s e_a^{(s)} \right)\\
    &\approx \omega^{(r_1)} - \sigma \sum_{a=1}^M e_a^{(r_1)}W_{aa} \left[\left(\alpha_{r_1}(t)e_a^{(r_1)} + \sum_{s\neq 0,r_1}\alpha^\infty_s e_a^{(s)} \right) - \frac{1}{6}\left(\alpha_{r_1}(t)e_a^{(r_1)} + \sum_{s\neq 0,r_1}\alpha^\infty_s e_a^{(s)} \right)^3\right]\\
    &\approx \omega^{(r_1)} - \sigma \sum_{a=1}^M e_a^{(r_1)}W_{aa} \left[\left(\alpha_{r_1}(t)e_a^{(r_1)} + \sum_{s\neq 0,r_1}\alpha^\infty_s e_a^{(s)} \right) - \frac{1}{6}\left(3(\alpha_{r_1}(t)e_a^{(r_1)})^2  \sum_{s\neq 0,r_1}\alpha^\infty_s e_a^{(s)} \right)\right]\\
    &= \omega^{(r_1)} - \sigma\lambda_{r_1}\alpha_{r_1}(t) - \sigma x_{r_1} \alpha_{r_1}^2(t)
\end{align*}

Where we define

\begin{align*}
    x_{r_1} &= \sum_{s\neq 0,r_1} \frac{\omega^{(s)}}{2 \sigma \lambda_s}\sum_{a=1}^M W_{aa} (e_a^{(r_1)})^3e_a^{(s)}.
\end{align*}

After this derivation, introducing weight into the system, the rest of the analysis follows the work of Kalloniatis\cite{kalloniatis_incoherence_2010}.

\section{SBMs approach AEPs in large n limit\label{app:SBMlimit}}

\begin{prop}
    With a choice of vertex blocks for an SBM model, there is a natural partition for any sampled graph. Thus a sampling induces corresponding $E=LP-PL^\pi$. In the limit of graph size, the values of $E$ approach zero.
\end{prop}

\begin{proof}
    We consider how $E$ scales with the size of the graph in the context of the SBM model. This allows us to scale a graph, while keeping the Quotient graph associated to the $\delta-$QEP we consider on that graph, constant in structure (though the edge weights are subject to change). Recall that the SBM model involves a partition of the vertices such that the connection probabilities between vertices in the graph model are determined solely by the partition blocks they are assigned to. Thus choose our SBM probability matrix to be 
    \begin{align*}
        Pr=
        \begin{bmatrix}
            \rho_{11} & \cdots & \rho_{1r}\\
            \vdots & & \vdots \\
            \rho_{r1} & \cdots & \rho_{rr}
        \end{bmatrix}
    \end{align*}
    This means that with some choice of vertex communities, $C_1,\dotsc, C_r$, we can sample an adjacency matrix.  
    \begin{align*}
        A &= 
        \begin{bmatrix}
            a_{11} & \cdots & a_{1n}\\
            \vdots & & \vdots \\
            a_{n1} & \cdots & a_{nn}
        \end{bmatrix}
        = 
        \begin{bmatrix}
            A^{11} & \cdots & A^{1r}\\
            \vdots & & \vdots \\
            A^{r1} & \cdots & A^{rr}
        \end{bmatrix}\\
    \end{align*}
    Where $A^{pq}$ denotes the the adjacency block describing connections between communities $p$ and $q$. Every entry in $A^{pq}$ with $p\neq q$ has probability $\rho_{pq}$ of being a $1$ and is otherwise zero. The diagonal blocks $A^{pp}$ have the additional requirements of being symmetric and maintaining zeros on their diagonal. We vary these communities while keeping $Pr$ constant. 
    
    This construction allows us to create an indicator matrix $P$ for the partition that the communities provide. This allows us to define $L$ and $L^\pi$. Now consider $LP$. This sums the columns of $L$ whose vertices belong to the same community. This means that each block $L^{pq}$ of $L$ has its rows summed together. Each of these entries correspond to the out degree of the vertex $i$ in community $C_p$ to the community $C_q$, which we denote as $d_{out}(i,C_q)$. The key observation to make here is that for each vertex $i$, the connection with every vertex in $C_q$ is sampled \textbf{independently} with probability $\rho_{pq}$. It follows that $\mathbb{E}(d_{out}(i,C_q))=\rho_{pq} |C_q|$. By the law of large numbers, the values will approach this expectation with larger sample sizes. This means that in the limit of scaling $C_q$, the value of $d_{out}(i,C_q)$ approaches $\rho_{pq} |C_q|$. This is true for all $i$ in $C_p$. 
    
    Let $d_{avg}(C_p,C_q)$ denote the average value of $d_{out}(i,C_q)$ over all $i\in C_p$. From our previous claim, $\mathbb{E}(d_{avg}(C_p,C_q))=\rho_{pq} |C_q|$ by linearity. Thus $\mathbb{E}(d_{avg}(C_p,C_q) - d_{out}(i,C_q))=0$, and consequently $\mathbb{E}(E) = \mathbb{E}(LP - PL^\pi) = 0$. Since the entries of $E$ are zero in expectation, By the law of large numbers, we can expect them to approach that value as the block sizes grow.

\end{proof}

We mention as a note that the above result is a consequence of the structure imposed by the SBM model. Namely that its sampled graphs are equitable partitions in expectation. This is as a result of their edge sampling within blocks being of uniform probability. We also note that the above relates to the work of Schaub et al. on stochastic AEPs.\cite{schaub_hierarchical_2023}

% \section*{References}
\bibliography{refs}

\begin{thebibliography}{50}%
\makeatletter
\providecommand \@ifxundefined [1]{%
 \@ifx{#1\undefined}
}%
\providecommand \@ifnum [1]{%
 \ifnum #1\expandafter \@firstoftwo
 \else \expandafter \@secondoftwo
 \fi
}%
\providecommand \@ifx [1]{%
 \ifx #1\expandafter \@firstoftwo
 \else \expandafter \@secondoftwo
 \fi
}%
\providecommand \natexlab [1]{#1}%
\providecommand \enquote  [1]{``#1''}%
\providecommand \bibnamefont  [1]{#1}%
\providecommand \bibfnamefont [1]{#1}%
\providecommand \citenamefont [1]{#1}%
\providecommand \href@noop [0]{\@secondoftwo}%
\providecommand \href [0]{\begingroup \@sanitize@url \@href}%
\providecommand \@href[1]{\@@startlink{#1}\@@href}%
\providecommand \@@href[1]{\endgroup#1\@@endlink}%
\providecommand \@sanitize@url [0]{\catcode `\\12\catcode `\$12\catcode `\&12\catcode `\#12\catcode `\^12\catcode `\_12\catcode `\%12\relax}%
\providecommand \@@startlink[1]{}%
\providecommand \@@endlink[0]{}%
\providecommand \url  [0]{\begingroup\@sanitize@url \@url }%
\providecommand \@url [1]{\endgroup\@href {#1}{\urlprefix }}%
\providecommand \urlprefix  [0]{URL }%
\providecommand \Eprint [0]{\href }%
\providecommand \doibase [0]{http://dx.doi.org/}%
\providecommand \selectlanguage [0]{\@gobble}%
\providecommand \bibinfo  [0]{\@secondoftwo}%
\providecommand \bibfield  [0]{\@secondoftwo}%
\providecommand \translation [1]{[#1]}%
\providecommand \BibitemOpen [0]{}%
\providecommand \bibitemStop [0]{}%
\providecommand \bibitemNoStop [0]{.\EOS\space}%
\providecommand \EOS [0]{\spacefactor3000\relax}%
\providecommand \BibitemShut  [1]{\csname bibitem#1\endcsname}%
\let\auto@bib@innerbib\@empty
%</preamble>
\bibitem [{\citenamefont {Atasoy}, \citenamefont {Donnelly},\ and\ \citenamefont {Pearson}(2016)}]{atasoy_human_2016}%
  \BibitemOpen
  \bibfield  {author} {\bibinfo {author} {\bibfnamefont {S.}~\bibnamefont {Atasoy}}, \bibinfo {author} {\bibfnamefont {I.}~\bibnamefont {Donnelly}}, \ and\ \bibinfo {author} {\bibfnamefont {J.}~\bibnamefont {Pearson}},\ }\bibfield  {title} {{\selectlanguage {english}\enquote {\bibinfo {title} {Human brain networks function in connectome-specific harmonic waves},}\ }}\href {\doibase 10.1038/ncomms10340} {\bibfield  {journal} {\bibinfo  {journal} {Nature Communications}\ }\textbf {\bibinfo {volume} {7}},\ \bibinfo {pages} {10340} (\bibinfo {year} {2016})},\ \bibinfo {note} {publisher: Nature Publishing Group}\BibitemShut {NoStop}%
\bibitem [{\citenamefont {Majhi}\ \emph {et~al.}(2025)\citenamefont {Majhi}, \citenamefont {Ghosh}, \citenamefont {Pal}, \citenamefont {Pal}, \citenamefont {Pal}, \citenamefont {Ghosh}, \citenamefont {Završnik},\ and\ \citenamefont {Perc}}]{majhi_patterns_2025}%
  \BibitemOpen
  \bibfield  {author} {\bibinfo {author} {\bibfnamefont {S.}~\bibnamefont {Majhi}}, \bibinfo {author} {\bibfnamefont {S.}~\bibnamefont {Ghosh}}, \bibinfo {author} {\bibfnamefont {P.~K.}\ \bibnamefont {Pal}}, \bibinfo {author} {\bibfnamefont {S.}~\bibnamefont {Pal}}, \bibinfo {author} {\bibfnamefont {T.~K.}\ \bibnamefont {Pal}}, \bibinfo {author} {\bibfnamefont {D.}~\bibnamefont {Ghosh}}, \bibinfo {author} {\bibfnamefont {J.}~\bibnamefont {Završnik}}, \ and\ \bibinfo {author} {\bibfnamefont {M.}~\bibnamefont {Perc}},\ }\bibfield  {title} {{\selectlanguage {english}\enquote {\bibinfo {title} {Patterns of neuronal synchrony in higher-order networks},}\ }}\href {\doibase 10.1016/j.plrev.2024.12.013} {\bibfield  {journal} {\bibinfo  {journal} {Physics of Life Reviews}\ }\textbf {\bibinfo {volume} {52}},\ \bibinfo {pages} {144--170} (\bibinfo {year} {2025})}\BibitemShut {NoStop}%
\bibitem [{\citenamefont {Kuramoto}(1984)}]{Kuramoto_chem_oscillations}%
  \BibitemOpen
  \bibfield  {author} {\bibinfo {author} {\bibfnamefont {Y.}~\bibnamefont {Kuramoto}},\ }\href {https://doi.org/10.1007/978-3-642-69689-3} {\emph {\bibinfo {title} {Chemical oscillations, waves, and turbulence}}},\ Vol.\ \bibinfo {volume} {VIII}\ (\bibinfo  {publisher} {Springer Berlin, Heidelberg},\ \bibinfo {year} {1984})\BibitemShut {NoStop}%
\bibitem [{\citenamefont {Strogatz}(2000)}]{strogatz_kuramoto_2000}%
  \BibitemOpen
  \bibfield  {author} {\bibinfo {author} {\bibfnamefont {S.~H.}\ \bibnamefont {Strogatz}},\ }\bibfield  {title} {\enquote {\bibinfo {title} {From {Kuramoto} to {Crawford}: exploring the onset of synchronization in populations of coupled oscillators},}\ }\href {\doibase 10.1016/S0167-2789(00)00094-4} {\bibfield  {journal} {\bibinfo  {journal} {Physica D: Nonlinear Phenomena}\ }\textbf {\bibinfo {volume} {143}},\ \bibinfo {pages} {1--20} (\bibinfo {year} {2000})}\BibitemShut {NoStop}%
\bibitem [{\citenamefont {Parastesh}\ \emph {et~al.}(2021)\citenamefont {Parastesh}, \citenamefont {Jafari}, \citenamefont {Azarnoush}, \citenamefont {Shahriari}, \citenamefont {Wang}, \citenamefont {Boccaletti},\ and\ \citenamefont {Perc}}]{parastesh2021chimeras}%
  \BibitemOpen
  \bibfield  {author} {\bibinfo {author} {\bibfnamefont {F.}~\bibnamefont {Parastesh}}, \bibinfo {author} {\bibfnamefont {S.}~\bibnamefont {Jafari}}, \bibinfo {author} {\bibfnamefont {H.}~\bibnamefont {Azarnoush}}, \bibinfo {author} {\bibfnamefont {Z.}~\bibnamefont {Shahriari}}, \bibinfo {author} {\bibfnamefont {Z.}~\bibnamefont {Wang}}, \bibinfo {author} {\bibfnamefont {S.}~\bibnamefont {Boccaletti}}, \ and\ \bibinfo {author} {\bibfnamefont {M.}~\bibnamefont {Perc}},\ }\bibfield  {title} {\enquote {\bibinfo {title} {Chimeras},}\ }\href@noop {} {\bibfield  {journal} {\bibinfo  {journal} {Physics Reports}\ }\textbf {\bibinfo {volume} {898}},\ \bibinfo {pages} {1--114} (\bibinfo {year} {2021})}\BibitemShut {NoStop}%
\bibitem [{\citenamefont {Skardal}, \citenamefont {Ott},\ and\ \citenamefont {Restrepo}(2011)}]{skardal_cluster_2011}%
  \BibitemOpen
  \bibfield  {author} {\bibinfo {author} {\bibfnamefont {P.~S.}\ \bibnamefont {Skardal}}, \bibinfo {author} {\bibfnamefont {E.}~\bibnamefont {Ott}}, \ and\ \bibinfo {author} {\bibfnamefont {J.~G.}\ \bibnamefont {Restrepo}},\ }\bibfield  {title} {{\selectlanguage {english}\enquote {\bibinfo {title} {Cluster synchrony in systems of coupled phase oscillators with higher-order coupling},}\ }}\href {\doibase 10.1103/PhysRevE.84.036208} {\bibfield  {journal} {\bibinfo  {journal} {Physical Review E}\ }\textbf {\bibinfo {volume} {84}},\ \bibinfo {pages} {036208} (\bibinfo {year} {2011})}\BibitemShut {NoStop}%
\bibitem [{\citenamefont {Gong}\ and\ \citenamefont {Pikovsky}(2019)}]{gong_low-dimensional_2019}%
  \BibitemOpen
  \bibfield  {author} {\bibinfo {author} {\bibfnamefont {C.~C.}\ \bibnamefont {Gong}}\ and\ \bibinfo {author} {\bibfnamefont {A.}~\bibnamefont {Pikovsky}},\ }\bibfield  {title} {{\selectlanguage {english}\enquote {\bibinfo {title} {Low-dimensional dynamics for higher-order harmonic, globally coupled phase-oscillator ensembles},}\ }}\href {\doibase 10.1103/PhysRevE.100.062210} {\bibfield  {journal} {\bibinfo  {journal} {Physical Review E}\ }\textbf {\bibinfo {volume} {100}},\ \bibinfo {pages} {062210} (\bibinfo {year} {2019})}\BibitemShut {NoStop}%
\bibitem [{\citenamefont {Nguy{\^e}n}, \citenamefont {Lee},\ and\ \citenamefont {Stolz}(2025)}]{nguyen2025communities}%
  \BibitemOpen
  \bibfield  {author} {\bibinfo {author} {\bibfnamefont {T.~J.}\ \bibnamefont {Nguy{\^e}n}}, \bibinfo {author} {\bibfnamefont {D.}~\bibnamefont {Lee}}, \ and\ \bibinfo {author} {\bibfnamefont {B.~J.}\ \bibnamefont {Stolz}},\ }\bibfield  {title} {\enquote {\bibinfo {title} {Communities in the kuramoto model: Dynamics and detection via path signatures},}\ }\href@noop {} {\bibfield  {journal} {\bibinfo  {journal} {arXiv preprint arXiv:2503.17546}\ } (\bibinfo {year} {2025})}\BibitemShut {NoStop}%
\bibitem [{\citenamefont {Schaub}\ \emph {et~al.}(2016)\citenamefont {Schaub}, \citenamefont {O'Clery}, \citenamefont {Billeh}, \citenamefont {Delvenne}, \citenamefont {Lambiotte},\ and\ \citenamefont {Barahona}}]{schaub_graph_2016}%
  \BibitemOpen
  \bibfield  {author} {\bibinfo {author} {\bibfnamefont {M.~T.}\ \bibnamefont {Schaub}}, \bibinfo {author} {\bibfnamefont {N.}~\bibnamefont {O'Clery}}, \bibinfo {author} {\bibfnamefont {Y.~N.}\ \bibnamefont {Billeh}}, \bibinfo {author} {\bibfnamefont {J.-C.}\ \bibnamefont {Delvenne}}, \bibinfo {author} {\bibfnamefont {R.}~\bibnamefont {Lambiotte}}, \ and\ \bibinfo {author} {\bibfnamefont {M.}~\bibnamefont {Barahona}},\ }\bibfield  {title} {\enquote {\bibinfo {title} {Graph partitions and cluster synchronization in networks of oscillators},}\ }\href {\doibase 10.1063/1.4961065} {\bibfield  {journal} {\bibinfo  {journal} {Chaos: An Interdisciplinary Journal of Nonlinear Science}\ }\textbf {\bibinfo {volume} {26}},\ \bibinfo {pages} {094821} (\bibinfo {year} {2016})}\BibitemShut {NoStop}%
\bibitem [{\citenamefont {Kato}\ and\ \citenamefont {Ishii}(2023)}]{kato_cluster_2023}%
  \BibitemOpen
  \bibfield  {author} {\bibinfo {author} {\bibfnamefont {R.}~\bibnamefont {Kato}}\ and\ \bibinfo {author} {\bibfnamefont {H.}~\bibnamefont {Ishii}},\ }\href {\doibase 8386-8401} {\enquote {\bibinfo {title} {Cluster {Synchronization} of {Kuramoto} {Oscillators} and the {Method} of {Averaging}},}\ } (\bibinfo {year} {2023})\BibitemShut {NoStop}%
\bibitem [{\citenamefont {Abiad}, \citenamefont {Hojny},\ and\ \citenamefont {Zeijlemaker}(2022)}]{abiad_characterizing_2022}%
  \BibitemOpen
  \bibfield  {author} {\bibinfo {author} {\bibfnamefont {A.}~\bibnamefont {Abiad}}, \bibinfo {author} {\bibfnamefont {C.}~\bibnamefont {Hojny}}, \ and\ \bibinfo {author} {\bibfnamefont {S.}~\bibnamefont {Zeijlemaker}},\ }\bibfield  {title} {\enquote {\bibinfo {title} {Characterizing and computing weight-equitable partitions of graphs},}\ }\href {\doibase 10.1016/j.laa.2022.03.003} {\bibfield  {journal} {\bibinfo  {journal} {Linear Algebra and its Applications}\ }\textbf {\bibinfo {volume} {645}},\ \bibinfo {pages} {30--51} (\bibinfo {year} {2022})}\BibitemShut {NoStop}%
\bibitem [{\citenamefont {O’Clery}\ \emph {et~al.}(2013)\citenamefont {O’Clery}, \citenamefont {Yuan}, \citenamefont {Stan},\ and\ \citenamefont {Barahona}}]{OClery_Observability_2013}%
  \BibitemOpen
  \bibfield  {author} {\bibinfo {author} {\bibfnamefont {N.}~\bibnamefont {O’Clery}}, \bibinfo {author} {\bibfnamefont {Y.}~\bibnamefont {Yuan}}, \bibinfo {author} {\bibfnamefont {G.-B.}\ \bibnamefont {Stan}}, \ and\ \bibinfo {author} {\bibfnamefont {M.}~\bibnamefont {Barahona}},\ }\bibfield  {title} {{\selectlanguage {english}\enquote {\bibinfo {title} {Observability and coarse graining of consensus dynamics through the external equitable partition},}\ }}\href {\doibase 10.1103/PhysRevE.88.042805} {\bibfield  {journal} {\bibinfo  {journal} {Physical Review E}\ }\textbf {\bibinfo {volume} {88}},\ \bibinfo {pages} {042805} (\bibinfo {year} {2013})}\BibitemShut {NoStop}%
\bibitem [{\citenamefont {Shuman}\ \emph {et~al.}(2013)\citenamefont {Shuman}, \citenamefont {Narang}, \citenamefont {Frossard}, \citenamefont {Ortega},\ and\ \citenamefont {Vandergheynst}}]{shuman_emerging_2013}%
  \BibitemOpen
  \bibfield  {author} {\bibinfo {author} {\bibfnamefont {D.~I.}\ \bibnamefont {Shuman}}, \bibinfo {author} {\bibfnamefont {S.~K.}\ \bibnamefont {Narang}}, \bibinfo {author} {\bibfnamefont {P.}~\bibnamefont {Frossard}}, \bibinfo {author} {\bibfnamefont {A.}~\bibnamefont {Ortega}}, \ and\ \bibinfo {author} {\bibfnamefont {P.}~\bibnamefont {Vandergheynst}},\ }\bibfield  {title} {\enquote {\bibinfo {title} {The emerging field of signal processing on graphs: {Extending} high-dimensional data analysis to networks and other irregular domains},}\ }\href {\doibase 10.1109/MSP.2012.2235192} {\bibfield  {journal} {\bibinfo  {journal} {IEEE Signal Processing Magazine}\ }\textbf {\bibinfo {volume} {30}},\ \bibinfo {pages} {83--98} (\bibinfo {year} {2013})}\BibitemShut {NoStop}%
\bibitem [{\citenamefont {Zhang}\ and\ \citenamefont {Sun}(2021)}]{zhang_almost_2021}%
  \BibitemOpen
  \bibfield  {author} {\bibinfo {author} {\bibfnamefont {X.}~\bibnamefont {Zhang}}\ and\ \bibinfo {author} {\bibfnamefont {J.}~\bibnamefont {Sun}},\ }\bibfield  {title} {\enquote {\bibinfo {title} {Almost equitable partitions and controllability of leader–follower multi-agent systems},}\ }\href {\doibase 10.1016/j.automatica.2021.109740} {\bibfield  {journal} {\bibinfo  {journal} {Automatica}\ }\textbf {\bibinfo {volume} {131}},\ \bibinfo {pages} {109740} (\bibinfo {year} {2021})}\BibitemShut {NoStop}%
\bibitem [{\citenamefont {Krzakala}\ \emph {et~al.}(2013)\citenamefont {Krzakala}, \citenamefont {Moore}, \citenamefont {Mossel}, \citenamefont {Neeman}, \citenamefont {Sly}, \citenamefont {Zdeborová},\ and\ \citenamefont {Zhang}}]{krzakala_spectral_2013}%
  \BibitemOpen
  \bibfield  {author} {\bibinfo {author} {\bibfnamefont {F.}~\bibnamefont {Krzakala}}, \bibinfo {author} {\bibfnamefont {C.}~\bibnamefont {Moore}}, \bibinfo {author} {\bibfnamefont {E.}~\bibnamefont {Mossel}}, \bibinfo {author} {\bibfnamefont {J.}~\bibnamefont {Neeman}}, \bibinfo {author} {\bibfnamefont {A.}~\bibnamefont {Sly}}, \bibinfo {author} {\bibfnamefont {L.}~\bibnamefont {Zdeborová}}, \ and\ \bibinfo {author} {\bibfnamefont {P.}~\bibnamefont {Zhang}},\ }\bibfield  {title} {\enquote {\bibinfo {title} {Spectral redemption: clustering sparse networks},}\ }\href {\doibase 10.1073/pnas.1312486110} {\bibfield  {journal} {\bibinfo  {journal} {Proceedings of the National Academy of Sciences}\ }\textbf {\bibinfo {volume} {110}},\ \bibinfo {pages} {20935--20940} (\bibinfo {year} {2013})}\BibitemShut {NoStop}%
\bibitem [{\citenamefont {Schaub}, \citenamefont {Li},\ and\ \citenamefont {Peel}(2023)}]{schaub_hierarchical_2023}%
  \BibitemOpen
  \bibfield  {author} {\bibinfo {author} {\bibfnamefont {M.~T.}\ \bibnamefont {Schaub}}, \bibinfo {author} {\bibfnamefont {J.}~\bibnamefont {Li}}, \ and\ \bibinfo {author} {\bibfnamefont {L.}~\bibnamefont {Peel}},\ }\bibfield  {title} {\enquote {\bibinfo {title} {Hierarchical community structure in networks},}\ }\href {\doibase 10.1103/PhysRevE.107.054305} {\bibfield  {journal} {\bibinfo  {journal} {Physical Review E}\ }\textbf {\bibinfo {volume} {107}},\ \bibinfo {pages} {054305} (\bibinfo {year} {2023})}\BibitemShut {NoStop}%
\bibitem [{\citenamefont {Thibeault}\ \emph {et~al.}(2020)\citenamefont {Thibeault}, \citenamefont {St-Onge}, \citenamefont {Dubé},\ and\ \citenamefont {Desrosiers}}]{thibeault_threefold_2020}%
  \BibitemOpen
  \bibfield  {author} {\bibinfo {author} {\bibfnamefont {V.}~\bibnamefont {Thibeault}}, \bibinfo {author} {\bibfnamefont {G.}~\bibnamefont {St-Onge}}, \bibinfo {author} {\bibfnamefont {L.~J.}\ \bibnamefont {Dubé}}, \ and\ \bibinfo {author} {\bibfnamefont {P.}~\bibnamefont {Desrosiers}},\ }\bibfield  {title} {\enquote {\bibinfo {title} {Threefold way to the dimension reduction of dynamics on networks: an application to synchronization},}\ }\href {\doibase 10.1103/PhysRevResearch.2.043215} {\bibfield  {journal} {\bibinfo  {journal} {Physical Review Research}\ }\textbf {\bibinfo {volume} {2}},\ \bibinfo {pages} {043215} (\bibinfo {year} {2020})}\BibitemShut {NoStop}%
\bibitem [{\citenamefont {Thibeault}, \citenamefont {Allard},\ and\ \citenamefont {Desrosiers}(2024)}]{thibeault_low-rank_2024}%
  \BibitemOpen
  \bibfield  {author} {\bibinfo {author} {\bibfnamefont {V.}~\bibnamefont {Thibeault}}, \bibinfo {author} {\bibfnamefont {A.}~\bibnamefont {Allard}}, \ and\ \bibinfo {author} {\bibfnamefont {P.}~\bibnamefont {Desrosiers}},\ }\bibfield  {title} {\enquote {\bibinfo {title} {The low-rank hypothesis of complex systems},}\ }\href {\doibase 10.1038/s41567-023-02303-0} {\bibfield  {journal} {\bibinfo  {journal} {Nature Physics}\ }\textbf {\bibinfo {volume} {20}},\ \bibinfo {pages} {294--302} (\bibinfo {year} {2024})}\BibitemShut {NoStop}%
\bibitem [{\citenamefont {Cho}, \citenamefont {Nishikawa},\ and\ \citenamefont {Motter}(2017)}]{cho_stable_2017}%
  \BibitemOpen
  \bibfield  {author} {\bibinfo {author} {\bibfnamefont {Y.~S.}\ \bibnamefont {Cho}}, \bibinfo {author} {\bibfnamefont {T.}~\bibnamefont {Nishikawa}}, \ and\ \bibinfo {author} {\bibfnamefont {A.~E.}\ \bibnamefont {Motter}},\ }\bibfield  {title} {{\selectlanguage {english}\enquote {\bibinfo {title} {Stable {Chimeras} and {Independently} {Synchronizable} {Clusters}},}\ }}\href {\doibase 10.1103/PhysRevLett.119.084101} {\bibfield  {journal} {\bibinfo  {journal} {Physical Review Letters}\ }\textbf {\bibinfo {volume} {119}},\ \bibinfo {pages} {084101} (\bibinfo {year} {2017})}\BibitemShut {NoStop}%
\bibitem [{\citenamefont {Zhang}\ and\ \citenamefont {Motter}(2020)}]{zhang_symmetry-independent_2020}%
  \BibitemOpen
  \bibfield  {author} {\bibinfo {author} {\bibfnamefont {Y.}~\bibnamefont {Zhang}}\ and\ \bibinfo {author} {\bibfnamefont {A.~E.}\ \bibnamefont {Motter}},\ }\bibfield  {title} {{\selectlanguage {english}\enquote {\bibinfo {title} {Symmetry-{Independent} {Stability} {Analysis} of {Synchronization} {Patterns}},}\ }}\href {\doibase 10.1137/19M127358X} {\bibfield  {journal} {\bibinfo  {journal} {SIAM Review}\ }\textbf {\bibinfo {volume} {62}},\ \bibinfo {pages} {817--836} (\bibinfo {year} {2020})}\BibitemShut {NoStop}%
\bibitem [{\citenamefont {Pecora}\ \emph {et~al.}(2014)\citenamefont {Pecora}, \citenamefont {Sorrentino}, \citenamefont {Hagerstrom}, \citenamefont {Murphy},\ and\ \citenamefont {Roy}}]{pecora_cluster_2014}%
  \BibitemOpen
  \bibfield  {author} {\bibinfo {author} {\bibfnamefont {L.~M.}\ \bibnamefont {Pecora}}, \bibinfo {author} {\bibfnamefont {F.}~\bibnamefont {Sorrentino}}, \bibinfo {author} {\bibfnamefont {A.~M.}\ \bibnamefont {Hagerstrom}}, \bibinfo {author} {\bibfnamefont {T.~E.}\ \bibnamefont {Murphy}}, \ and\ \bibinfo {author} {\bibfnamefont {R.}~\bibnamefont {Roy}},\ }\bibfield  {title} {{\selectlanguage {english}\enquote {\bibinfo {title} {Cluster synchronization and isolated desynchronization in complex networks with symmetries},}\ }}\href {\doibase 10.1038/ncomms5079} {\bibfield  {journal} {\bibinfo  {journal} {Nature Communications}\ }\textbf {\bibinfo {volume} {5}},\ \bibinfo {pages} {4079} (\bibinfo {year} {2014})}\BibitemShut {NoStop}%
\bibitem [{\citenamefont {Khanra}\ \emph {et~al.}(2022)\citenamefont {Khanra}, \citenamefont {Ghosh}, \citenamefont {Alfaro-Bittner}, \citenamefont {Kundu}, \citenamefont {Boccaletti}, \citenamefont {Hens},\ and\ \citenamefont {Pal}}]{khanra_identifying_2022}%
  \BibitemOpen
  \bibfield  {author} {\bibinfo {author} {\bibfnamefont {P.}~\bibnamefont {Khanra}}, \bibinfo {author} {\bibfnamefont {S.}~\bibnamefont {Ghosh}}, \bibinfo {author} {\bibfnamefont {K.}~\bibnamefont {Alfaro-Bittner}}, \bibinfo {author} {\bibfnamefont {P.}~\bibnamefont {Kundu}}, \bibinfo {author} {\bibfnamefont {S.}~\bibnamefont {Boccaletti}}, \bibinfo {author} {\bibfnamefont {C.}~\bibnamefont {Hens}}, \ and\ \bibinfo {author} {\bibfnamefont {P.}~\bibnamefont {Pal}},\ }\bibfield  {title} {\enquote {\bibinfo {title} {Identifying symmetries and predicting cluster synchronization in complex networks},}\ }\href {\doibase 10.1016/j.chaos.2021.111703} {\bibfield  {journal} {\bibinfo  {journal} {Chaos, Solitons \& Fractals}\ }\textbf {\bibinfo {volume} {155}},\ \bibinfo {pages} {111703} (\bibinfo {year} {2022})},\ \bibinfo {note} {arXiv:2102.06957 [nlin]}\BibitemShut {NoStop}%
\bibitem [{\citenamefont {Newman}(2018)}]{Newman2018-ut}%
  \BibitemOpen
  \bibfield  {author} {\bibinfo {author} {\bibfnamefont {M.}~\bibnamefont {Newman}},\ }\href@noop {} {\emph {\bibinfo {title} {Networks}}},\ \bibinfo {edition} {2nd}\ ed.\ (\bibinfo  {publisher} {Oxford University Press},\ \bibinfo {address} {London, England},\ \bibinfo {year} {2018})\BibitemShut {NoStop}%
\bibitem [{\citenamefont {Brouwer}\ and\ \citenamefont {Haemers}(2012)}]{brouwer2012spectra}%
  \BibitemOpen
  \bibfield  {author} {\bibinfo {author} {\bibfnamefont {A.~E.}\ \bibnamefont {Brouwer}}\ and\ \bibinfo {author} {\bibfnamefont {W.~H.}\ \bibnamefont {Haemers}},\ }\href@noop {} {\emph {\bibinfo {title} {Spectra of Graphs}}}\ (\bibinfo  {publisher} {Springer},\ \bibinfo {year} {2012})\BibitemShut {NoStop}%
\bibitem [{\citenamefont {Cardoso}, \citenamefont {Delorme},\ and\ \citenamefont {Rama}(2007)}]{cardoso_Laplacian_2007}%
  \BibitemOpen
  \bibfield  {author} {\bibinfo {author} {\bibfnamefont {D.~M.}\ \bibnamefont {Cardoso}}, \bibinfo {author} {\bibfnamefont {C.}~\bibnamefont {Delorme}}, \ and\ \bibinfo {author} {\bibfnamefont {P.}~\bibnamefont {Rama}},\ }\bibfield  {title} {\enquote {\bibinfo {title} {Laplacian eigenvectors and eigenvalues and almost equitable partitions},}\ }\href {\doibase 10.1016/j.ejc.2005.03.006} {\bibfield  {journal} {\bibinfo  {journal} {European Journal of Combinatorics}\ }\textbf {\bibinfo {volume} {28}},\ \bibinfo {pages} {665--673} (\bibinfo {year} {2007})}\BibitemShut {NoStop}%
\bibitem [{\citenamefont {Aguilar}\ and\ \citenamefont {Gharesifard}(2016)}]{aguilar_almost_2016}%
  \BibitemOpen
  \bibfield  {author} {\bibinfo {author} {\bibfnamefont {C.~O.}\ \bibnamefont {Aguilar}}\ and\ \bibinfo {author} {\bibfnamefont {B.}~\bibnamefont {Gharesifard}},\ }\bibfield  {title} {\enquote {\bibinfo {title} {On almost equitable partitions and network controllability},}\ }in\ \href {\doibase 10.1109/ACC.2016.7524912} {\emph {\bibinfo {booktitle} {2016 {American} {Control} {Conference} ({ACC})}}}\ (\bibinfo  {publisher} {IEEE},\ \bibinfo {address} {Boston, MA, USA},\ \bibinfo {year} {2016})\ pp.\ \bibinfo {pages} {179--184}\BibitemShut {NoStop}%
\bibitem [{\citenamefont {Penrose}(1955)}]{penrose1955generalized}%
  \BibitemOpen
  \bibfield  {author} {\bibinfo {author} {\bibfnamefont {R.}~\bibnamefont {Penrose}},\ }\bibfield  {title} {\enquote {\bibinfo {title} {A generalized inverse for matrices},}\ }in\ \href@noop {} {\emph {\bibinfo {booktitle} {Mathematical proceedings of the Cambridge philosophical society}}},\ Vol.~\bibinfo {volume} {51}\ (\bibinfo {organization} {Cambridge University Press},\ \bibinfo {year} {1955})\ pp.\ \bibinfo {pages} {406--413}\BibitemShut {NoStop}%
\bibitem [{\citenamefont {Haberman}(2013)}]{haberman2013applied}%
  \BibitemOpen
  \bibfield  {author} {\bibinfo {author} {\bibfnamefont {R.}~\bibnamefont {Haberman}},\ }\href {https://books.google.com/books?id=hGNwLgEACAAJ} {\emph {\bibinfo {title} {Applied Partial Differential Equations: With Fourier Series and Boundary Value Problems}}},\ Featured Titles for Partial Differential Equations\ (\bibinfo  {publisher} {Pearson},\ \bibinfo {year} {2013})\BibitemShut {NoStop}%
\bibitem [{\citenamefont {Belkin}\ and\ \citenamefont {Niyogi}(2008)}]{belkin2008towards}%
  \BibitemOpen
  \bibfield  {author} {\bibinfo {author} {\bibfnamefont {M.}~\bibnamefont {Belkin}}\ and\ \bibinfo {author} {\bibfnamefont {P.}~\bibnamefont {Niyogi}},\ }\bibfield  {title} {\enquote {\bibinfo {title} {Towards a theoretical foundation for laplacian-based manifold methods},}\ }\href@noop {} {\bibfield  {journal} {\bibinfo  {journal} {Journal of Computer and System Sciences}\ }\textbf {\bibinfo {volume} {74}},\ \bibinfo {pages} {1289--1308} (\bibinfo {year} {2008})}\BibitemShut {NoStop}%
\bibitem [{\citenamefont {Seabrook}\ and\ \citenamefont {Wiskott}(2023)}]{Seabrooketal_Markov_Spectral}%
  \BibitemOpen
  \bibfield  {author} {\bibinfo {author} {\bibfnamefont {E.}~\bibnamefont {Seabrook}}\ and\ \bibinfo {author} {\bibfnamefont {L.}~\bibnamefont {Wiskott}},\ }\bibfield  {title} {\enquote {\bibinfo {title} {A tutorial on the spectral theory of markov chains},}\ }\href {\doibase 10.1162/neco_a_01611} {\bibfield  {journal} {\bibinfo  {journal} {Neural Computation}\ }\textbf {\bibinfo {volume} {35}},\ \bibinfo {pages} {1713--1796} (\bibinfo {year} {2023})},\ \Eprint {http://arxiv.org/abs/https://direct.mit.edu/neco/article-pdf/35/11/1713/2335627/neco\_a\_01611.pdf} {https://direct.mit.edu/neco/article-pdf/35/11/1713/2335627/neco\_a\_01611.pdf} \BibitemShut {NoStop}%
\bibitem [{\citenamefont {Boukrab}\ and\ \citenamefont {Pagès-Zamora}(2021)}]{boukrab_random-walk_2021}%
  \BibitemOpen
  \bibfield  {author} {\bibinfo {author} {\bibfnamefont {R.}~\bibnamefont {Boukrab}}\ and\ \bibinfo {author} {\bibfnamefont {A.}~\bibnamefont {Pagès-Zamora}},\ }\bibfield  {title} {\enquote {\bibinfo {title} {Random-{Walk} {Laplacian} for {Frequency} {Analysis} in {Periodic} {Graphs}},}\ }\href {\doibase 10.3390/s21041275} {\bibfield  {journal} {\bibinfo  {journal} {Sensors}\ }\textbf {\bibinfo {volume} {21}},\ \bibinfo {pages} {1275} (\bibinfo {year} {2021})}\BibitemShut {NoStop}%
\bibitem [{\citenamefont {Saade}, \citenamefont {Krzakala},\ and\ \citenamefont {Zdeborová}(2014)}]{saade_spectral_nodate}%
  \BibitemOpen
  \bibfield  {author} {\bibinfo {author} {\bibfnamefont {A.}~\bibnamefont {Saade}}, \bibinfo {author} {\bibfnamefont {F.}~\bibnamefont {Krzakala}}, \ and\ \bibinfo {author} {\bibfnamefont {L.}~\bibnamefont {Zdeborová}},\ }\bibfield  {title} {\enquote {\bibinfo {title} {Spectral {Clustering} of graphs with the {Bethe} {Hessian}},}\ }\href {\doibase 1406.1880} {\  (\bibinfo {year} {2014}),\ 1406.1880}\BibitemShut {NoStop}%
\bibitem [{\citenamefont {Kalloniatis}(2010)}]{kalloniatis_incoherence_2010}%
  \BibitemOpen
  \bibfield  {author} {\bibinfo {author} {\bibfnamefont {A.~C.}\ \bibnamefont {Kalloniatis}},\ }\bibfield  {title} {\enquote {\bibinfo {title} {From incoherence to synchronicity in the network {Kuramoto} model},}\ }\href {\doibase 10.1103/PhysRevE.82.066202} {\bibfield  {journal} {\bibinfo  {journal} {Physical Review E}\ }\textbf {\bibinfo {volume} {82}},\ \bibinfo {pages} {066202} (\bibinfo {year} {2010})}\BibitemShut {NoStop}%
\bibitem [{\citenamefont {Menara}\ \emph {et~al.}(2020)\citenamefont {Menara}, \citenamefont {Baggio}, \citenamefont {Bassett},\ and\ \citenamefont {Pasqualetti}}]{menara_stability_2020}%
  \BibitemOpen
  \bibfield  {author} {\bibinfo {author} {\bibfnamefont {T.}~\bibnamefont {Menara}}, \bibinfo {author} {\bibfnamefont {G.}~\bibnamefont {Baggio}}, \bibinfo {author} {\bibfnamefont {D.~S.}\ \bibnamefont {Bassett}}, \ and\ \bibinfo {author} {\bibfnamefont {F.}~\bibnamefont {Pasqualetti}},\ }\bibfield  {title} {\enquote {\bibinfo {title} {Stability {Conditions} for {Cluster} {Synchronization} in {Networks} of {Heterogeneous} {Kuramoto} {Oscillators}},}\ }\href {\doibase 10.1109/TCNS.2019.2903914} {\bibfield  {journal} {\bibinfo  {journal} {IEEE Transactions on Control of Network Systems}\ }\textbf {\bibinfo {volume} {7}},\ \bibinfo {pages} {302--314} (\bibinfo {year} {2020})}\BibitemShut {NoStop}%
\bibitem [{\citenamefont {Fan}\ and\ \citenamefont {Wang}(2023)}]{fan_eigenvector-based_2023}%
  \BibitemOpen
  \bibfield  {author} {\bibinfo {author} {\bibfnamefont {H.}~\bibnamefont {Fan}}\ and\ \bibinfo {author} {\bibfnamefont {X.}~\bibnamefont {Wang}},\ }\bibfield  {title} {\enquote {\bibinfo {title} {Eigenvector-based analysis of cluster synchronization in general complex networks of coupled chaotic oscillators},}\ }\href {\doibase 10.1007/s11467-023-1324-0} {\bibfield  {journal} {\bibinfo  {journal} {Frontiers of Physics}\ }\textbf {\bibinfo {volume} {18}},\ \bibinfo {pages} {45302} (\bibinfo {year} {2023})},\ \bibinfo {note} {arXiv:2103.01004 [nlin, physics:physics]}\BibitemShut {NoStop}%
\bibitem [{\citenamefont {De~Domenico}(2017)}]{de_domenico_diffusion_2017}%
  \BibitemOpen
  \bibfield  {author} {\bibinfo {author} {\bibfnamefont {M.}~\bibnamefont {De~Domenico}},\ }\bibfield  {title} {\enquote {\bibinfo {title} {Diffusion geometry unravels the emergence of functional clusters in collective phenomena},}\ }\href {\doibase 10.1103/PhysRevLett.118.168301} {\bibfield  {journal} {\bibinfo  {journal} {Physical Review Letters}\ }\textbf {\bibinfo {volume} {118}},\ \bibinfo {pages} {168301} (\bibinfo {year} {2017})}\BibitemShut {NoStop}%
\bibitem [{\citenamefont {Arenas}, \citenamefont {Díaz-Guilera},\ and\ \citenamefont {Pérez-Vicente}(2006)}]{arenas_synchronization_2006}%
  \BibitemOpen
  \bibfield  {author} {\bibinfo {author} {\bibfnamefont {A.}~\bibnamefont {Arenas}}, \bibinfo {author} {\bibfnamefont {A.}~\bibnamefont {Díaz-Guilera}}, \ and\ \bibinfo {author} {\bibfnamefont {C.~J.}\ \bibnamefont {Pérez-Vicente}},\ }\bibfield  {title} {\enquote {\bibinfo {title} {Synchronization {Reveals} {Topological} {Scales} in {Complex} {Networks}},}\ }\href {\doibase 10.1103/PhysRevLett.96.114102} {\bibfield  {journal} {\bibinfo  {journal} {Physical Review Letters}\ }\textbf {\bibinfo {volume} {96}},\ \bibinfo {pages} {114102} (\bibinfo {year} {2006})}\BibitemShut {NoStop}%
\bibitem [{\citenamefont {Scholkemper}\ and\ \citenamefont {Schaub}(2023)}]{scholkemper_optimization-based_nodate}%
  \BibitemOpen
  \bibfield  {author} {\bibinfo {author} {\bibfnamefont {M.}~\bibnamefont {Scholkemper}}\ and\ \bibinfo {author} {\bibfnamefont {M.~T.}\ \bibnamefont {Schaub}},\ }\bibfield  {title} {\enquote {\bibinfo {title} {An optimization-based approach to node role discovery in networks: Approximating equitable partitions},}\ }in\ \href {https://proceedings.neurips.cc/paper_files/paper/2023/file/e1c73e9595126794186536cfbbed012f-Paper-Conference.pdf} {\emph {\bibinfo {booktitle} {Advances in Neural Information Processing Systems}}},\ Vol.~\bibinfo {volume} {36},\ \bibinfo {editor} {edited by\ \bibinfo {editor} {\bibfnamefont {A.}~\bibnamefont {Oh}}, \bibinfo {editor} {\bibfnamefont {T.}~\bibnamefont {Naumann}}, \bibinfo {editor} {\bibfnamefont {A.}~\bibnamefont {Globerson}}, \bibinfo {editor} {\bibfnamefont {K.}~\bibnamefont {Saenko}}, \bibinfo {editor} {\bibfnamefont {M.}~\bibnamefont {Hardt}}, \ and\ \bibinfo {editor} {\bibfnamefont {S.}~\bibnamefont {Levine}}}\ (\bibinfo  {publisher} {Curran Associates, Inc.},\
  \bibinfo {year} {2023})\ pp.\ \bibinfo {pages} {71358--71374}\BibitemShut {NoStop}%
\bibitem [{\citenamefont {Duan}\ and\ \citenamefont {Mai}(2020)}]{duan_spectral_2020}%
  \BibitemOpen
  \bibfield  {author} {\bibinfo {author} {\bibfnamefont {L.}~\bibnamefont {Duan}}\ and\ \bibinfo {author} {\bibfnamefont {X.}~\bibnamefont {Mai}},\ }\bibfield  {title} {\enquote {\bibinfo {title} {Spectral clustering-based resting-state network detection approach for functional near-infrared spectroscopy},}\ }\href {\doibase 10.1364/BOE.387919} {\bibfield  {journal} {\bibinfo  {journal} {Biomedical Optics Express}\ }\textbf {\bibinfo {volume} {11}},\ \bibinfo {pages} {2191--2204} (\bibinfo {year} {2020})}\BibitemShut {NoStop}%
\bibitem [{\citenamefont {Patania}, \citenamefont {Allard},\ and\ \citenamefont {Young}(2023)}]{patania_exact_2023}%
  \BibitemOpen
  \bibfield  {author} {\bibinfo {author} {\bibfnamefont {A.}~\bibnamefont {Patania}}, \bibinfo {author} {\bibfnamefont {A.}~\bibnamefont {Allard}}, \ and\ \bibinfo {author} {\bibfnamefont {J.-G.}\ \bibnamefont {Young}},\ }\bibfield  {title} {\enquote {\bibinfo {title} {Exact and rapid linear clustering of networks with dynamic programming},}\ }\href {\doibase 10.1098/rspa.2023.0159} {\bibfield  {journal} {\bibinfo  {journal} {Proceedings of the Royal Society A: Mathematical, Physical and Engineering Sciences}\ }\textbf {\bibinfo {volume} {479}},\ \bibinfo {pages} {20230159} (\bibinfo {year} {2023})}\BibitemShut {NoStop}%
\bibitem [{\citenamefont {Deco}, \citenamefont {Jirsa},\ and\ \citenamefont {McIntosh}(2011)}]{deco_emerging_2011}%
  \BibitemOpen
  \bibfield  {author} {\bibinfo {author} {\bibfnamefont {G.}~\bibnamefont {Deco}}, \bibinfo {author} {\bibfnamefont {V.~K.}\ \bibnamefont {Jirsa}}, \ and\ \bibinfo {author} {\bibfnamefont {A.~R.}\ \bibnamefont {McIntosh}},\ }\bibfield  {title} {\enquote {\bibinfo {title} {Emerging concepts for the dynamical organization of resting-state activity in the brain},}\ }\href {\doibase 10.1038/nrn2961} {\bibfield  {journal} {\bibinfo  {journal} {Nature Reviews Neuroscience}\ }\textbf {\bibinfo {volume} {12}},\ \bibinfo {pages} {43--56} (\bibinfo {year} {2011})}\BibitemShut {NoStop}%
\bibitem [{\citenamefont {Cabral}\ \emph {et~al.}(2011)\citenamefont {Cabral}, \citenamefont {Hugues}, \citenamefont {Sporns},\ and\ \citenamefont {Deco}}]{cabral2011role}%
  \BibitemOpen
  \bibfield  {author} {\bibinfo {author} {\bibfnamefont {J.}~\bibnamefont {Cabral}}, \bibinfo {author} {\bibfnamefont {E.}~\bibnamefont {Hugues}}, \bibinfo {author} {\bibfnamefont {O.}~\bibnamefont {Sporns}}, \ and\ \bibinfo {author} {\bibfnamefont {G.}~\bibnamefont {Deco}},\ }\bibfield  {title} {\enquote {\bibinfo {title} {Role of local network oscillations in resting-state functional connectivity},}\ }\href@noop {} {\bibfield  {journal} {\bibinfo  {journal} {Neuroimage}\ }\textbf {\bibinfo {volume} {57}},\ \bibinfo {pages} {130--139} (\bibinfo {year} {2011})}\BibitemShut {NoStop}%
\bibitem [{\citenamefont {Faskowitz}\ \emph {et~al.}(2020)\citenamefont {Faskowitz}, \citenamefont {Esfahlani}, \citenamefont {Jo}, \citenamefont {Sporns},\ and\ \citenamefont {Betzel}}]{faskowitz2020edge}%
  \BibitemOpen
  \bibfield  {author} {\bibinfo {author} {\bibfnamefont {J.}~\bibnamefont {Faskowitz}}, \bibinfo {author} {\bibfnamefont {F.~Z.}\ \bibnamefont {Esfahlani}}, \bibinfo {author} {\bibfnamefont {Y.}~\bibnamefont {Jo}}, \bibinfo {author} {\bibfnamefont {O.}~\bibnamefont {Sporns}}, \ and\ \bibinfo {author} {\bibfnamefont {R.~F.}\ \bibnamefont {Betzel}},\ }\bibfield  {title} {\enquote {\bibinfo {title} {Edge-centric functional network representations of human cerebral cortex reveal overlapping system-level architecture},}\ }\href@noop {} {\bibfield  {journal} {\bibinfo  {journal} {Nature neuroscience}\ }\textbf {\bibinfo {volume} {23}},\ \bibinfo {pages} {1644--1654} (\bibinfo {year} {2020})}\BibitemShut {NoStop}%
\bibitem [{\citenamefont {Pope}\ \emph {et~al.}(2023)\citenamefont {Pope}, \citenamefont {Seguin}, \citenamefont {Varley}, \citenamefont {Faskowitz},\ and\ \citenamefont {Sporns}}]{pope2023co}%
  \BibitemOpen
  \bibfield  {author} {\bibinfo {author} {\bibfnamefont {M.}~\bibnamefont {Pope}}, \bibinfo {author} {\bibfnamefont {C.}~\bibnamefont {Seguin}}, \bibinfo {author} {\bibfnamefont {T.~F.}\ \bibnamefont {Varley}}, \bibinfo {author} {\bibfnamefont {J.}~\bibnamefont {Faskowitz}}, \ and\ \bibinfo {author} {\bibfnamefont {O.}~\bibnamefont {Sporns}},\ }\bibfield  {title} {\enquote {\bibinfo {title} {Co-evolving dynamics and topology in a coupled oscillator model of resting brain function},}\ }\href@noop {} {\bibfield  {journal} {\bibinfo  {journal} {NeuroImage}\ }\textbf {\bibinfo {volume} {277}},\ \bibinfo {pages} {120266} (\bibinfo {year} {2023})}\BibitemShut {NoStop}%
\bibitem [{\citenamefont {Chow}(2013)}]{Chow_Power_Grid_Coherency}%
  \BibitemOpen
  \bibfield  {author} {\bibinfo {author} {\bibfnamefont {J.~H.}\ \bibnamefont {Chow}},\ }\href {https://link.springer.com/book/10.1007/978-1-4614-1803-03} {\emph {\bibinfo {title} {Power System Coherency and Model Reduction}}},\ Vol.~\bibinfo {volume} {I}\ (\bibinfo  {publisher} {Springer New York, NY},\ \bibinfo {year} {2013})\BibitemShut {NoStop}%
\bibitem [{\citenamefont {Dorfler}\ and\ \citenamefont {Bullo}(2012)}]{dorfler2012synchronization}%
  \BibitemOpen
  \bibfield  {author} {\bibinfo {author} {\bibfnamefont {F.}~\bibnamefont {Dorfler}}\ and\ \bibinfo {author} {\bibfnamefont {F.}~\bibnamefont {Bullo}},\ }\bibfield  {title} {\enquote {\bibinfo {title} {Synchronization and transient stability in power networks and nonuniform kuramoto oscillators},}\ }\href@noop {} {\bibfield  {journal} {\bibinfo  {journal} {SIAM Journal on Control and Optimization}\ }\textbf {\bibinfo {volume} {50}},\ \bibinfo {pages} {1616--1642} (\bibinfo {year} {2012})}\BibitemShut {NoStop}%
\bibitem [{\citenamefont {Motter}\ \emph {et~al.}(2013)\citenamefont {Motter}, \citenamefont {Myers}, \citenamefont {Anghel},\ and\ \citenamefont {Nishikawa}}]{motter2013spontaneous}%
  \BibitemOpen
  \bibfield  {author} {\bibinfo {author} {\bibfnamefont {A.~E.}\ \bibnamefont {Motter}}, \bibinfo {author} {\bibfnamefont {S.~A.}\ \bibnamefont {Myers}}, \bibinfo {author} {\bibfnamefont {M.}~\bibnamefont {Anghel}}, \ and\ \bibinfo {author} {\bibfnamefont {T.}~\bibnamefont {Nishikawa}},\ }\bibfield  {title} {\enquote {\bibinfo {title} {Spontaneous synchrony in power-grid networks},}\ }\href@noop {} {\bibfield  {journal} {\bibinfo  {journal} {Nature Physics}\ }\textbf {\bibinfo {volume} {9}},\ \bibinfo {pages} {191--197} (\bibinfo {year} {2013})}\BibitemShut {NoStop}%
\bibitem [{\citenamefont {Nishikawa}\ and\ \citenamefont {Motter}(2015)}]{nishikawa2015comparative}%
  \BibitemOpen
  \bibfield  {author} {\bibinfo {author} {\bibfnamefont {T.}~\bibnamefont {Nishikawa}}\ and\ \bibinfo {author} {\bibfnamefont {A.~E.}\ \bibnamefont {Motter}},\ }\bibfield  {title} {\enquote {\bibinfo {title} {Comparative analysis of existing models for power-grid synchronization},}\ }\href@noop {} {\bibfield  {journal} {\bibinfo  {journal} {New Journal of Physics}\ }\textbf {\bibinfo {volume} {17}},\ \bibinfo {pages} {015012} (\bibinfo {year} {2015})}\BibitemShut {NoStop}%
\bibitem [{\citenamefont {Timofeyev}(2025)}]{timofeyev2025spectral}%
  \BibitemOpen
  \bibfield  {author} {\bibinfo {author} {\bibfnamefont {T.}~\bibnamefont {Timofeyev}},\ }\href@noop {} {\enquote {\bibinfo {title} {Spectral kuramoto},}\ }\bibinfo {howpublished} {\url{https://github.com/TobiasTimofeyev/Spectral_Kuramoto}} (\bibinfo {year} {2025})\BibitemShut {NoStop}%
\bibitem [{\citenamefont {Breakspear}, \citenamefont {Heitmann},\ and\ \citenamefont {Daffertshofer}(2010)}]{breakspear_generative_2010}%
  \BibitemOpen
  \bibfield  {author} {\bibinfo {author} {\bibfnamefont {M.}~\bibnamefont {Breakspear}}, \bibinfo {author} {\bibfnamefont {S.}~\bibnamefont {Heitmann}}, \ and\ \bibinfo {author} {\bibfnamefont {A.}~\bibnamefont {Daffertshofer}},\ }\bibfield  {title} {\enquote {\bibinfo {title} {Generative {Models} of {Cortical} {Oscillations}: {Neurobiological} {Implications} of the {Kuramoto} {Model}},}\ }\href {\doibase 10.3389/fnhum.2010.00190} {\bibfield  {journal} {\bibinfo  {journal} {Frontiers in Human Neuroscience}\ }\textbf {\bibinfo {volume} {4}} (\bibinfo {year} {2010}),\ 10.3389/fnhum.2010.00190}\BibitemShut {NoStop}%
\end{thebibliography}%

\end{document}